\documentclass[12pt, final]{l4dc2023}

\title[Interpreting Constrained MARL]{Interpreting Primal-Dual Algorithms for Constrained Multiagent Reinforcement Learning}
\usepackage{times}
\usepackage{xcolor,wrapfig,amssymb,amsmath,scalerel,ifthen}
\usepackage[format=plain,size=small]{caption}
\usepackage{algorithmic}
\usepackage{algorithm}

\DeclareMathOperator*{\G}{\scaleobj{0.85}{\scalerel*{\Gamma}{\sum}}}

\newcommand{\N}{\mathbb{N}}
\newcommand{\R}{\mathbb{R}}
\renewcommand{\S}{\mathcal{S}}
\newcommand{\n}{\mathcal{N}}
\newcommand{\A}{\mathcal{A}}
\newcommand{\X}{\mathcal{X}}
\newcommand{\U}{\mathcal{U}}
\newcommand{\E}{\mathbb{E}}
\newcommand{\M}{\mathcal{M}}
\newcommand{\D}{\mathcal{D}}

\newtheorem{assumption}[theorem]{Assumption}

\newboolean{versiontwo}
\setboolean{versiontwo}{true}
\newcommand{\newversion}[2]{\ifthenelse{\boolean{versiontwo}}{#1}{#2}}

\newcommand\m[1]{\begin{bmatrix}#1\end{bmatrix}} 

\newboolean{fullversion}
\setboolean{fullversion}{true}
\newcommand{\includeappendix}[2]{\ifthenelse{\boolean{fullversion}}{#1}{#2}}

\author{%
 \Name{Daniel Tabas} \Email{dtabas@uw.edu}\\
 \addr University of Washington Department of Electrical Engineering, Seattle, WA, USA%
 \AND
 \Name{Ahmed S. Zamzam} \Email{ahmed.zamzam@nrel.gov}\\
 \addr National Renewable Energy Laboratory, Golden, CO, USA%
 \AND
 \Name{Baosen Zhang} \Email{zhangbao@uw.edu}\\
 \addr University of Washington Department of Electrical Engineering, Seattle, WA, USA
}

\begin{document}

\maketitle

\begin{abstract}%
Constrained multiagent reinforcement learning (C-MARL) is gaining importance as MARL algorithms find new applications in real-world systems ranging from energy systems to drone swarms. Most C-MARL algorithms use a primal-dual approach to enforce constraints through a penalty function added to the reward. In this paper, we study the structural effects of this penalty term on the MARL problem. First, we show that the standard practice of using the constraint function as the penalty leads to a weak notion of safety. However, by making simple modifications to the penalty term, we can enforce meaningful probabilistic (chance and conditional value at risk) constraints. Second, we quantify the effect of the penalty term on the value function, uncovering an improved value estimation procedure. We use these insights to propose a constrained multiagent advantage actor critic (C-MAA2C) algorithm. Simulations in a simple constrained multiagent environment affirm that our reinterpretation of the primal-dual method in terms of probabilistic constraints is effective, and that our proposed value estimate accelerates convergence to a safe joint policy.
\end{abstract}


\section{Introduction}
\newversion{As reinforcement learning (RL) algorithms progress from virtual to cyber-physical applications, it will be necessary to address the challenges of safety, especially when systems are controlled by multiple agents. 
Examples of multiagent safety-critical systems include power grids \citep{Cui2022}, building energy management (BEM) systems \citep{Biagioni2022}, autonomous vehicle navigation \citep{Zhou2022a}, and drone swarms \citep{Chen2020}. In each of these applications, agents must learn to operate in a complicated environment while satisfying various local and system-wide constraints. Such constraints, derived from domain-specific knowledge, are designed to prevent damage to equipment, humans, or infrastructure or to preclude failure to complete some task or objective. 

Constrained multiagent reinforcement learning (C-MARL) poses challenges beyond the single-agent constrained reinforcement learning (C-RL) problem because the interactions between agents can influence both the satisfaction of constraints and the convergence of policies. The potential scale of C-MARL problems eliminates the possibility of directly using common model-based methods for C-RL, such as in \cite{Chen2018c,Ma2021a,Tabas2022}. The main  strategy for tackling C-MARL problems found in the literature is the Lagrangian or primal-dual method (see, e.g.~\cite{Lu2021a,Li2020,Lee2018,Parnika2021} and the references therein). Our aim is to understand some potential drawbacks of this approach and some ways these drawbacks can be mitigated. 

In the primal-dual approach to C-MARL, each agent receives a reward signal that is augmented with a penalty term designed to incentivize constraint satisfaction. The magnitude of the penalty term is tuned to steer policies away from constraint violations while not unnecessarily overshadowing the original reward. Although this approach has been shown to converge to a safe joint policy under certain assumptions
\citep{Lu2021a}, it changes the structure of the problem in a way that is not well understood, leading to two challenges.

The first challenge is that the primal-dual algorithm only enforces \emph{discounted sum constraints} derived from the original safety constraints of the system. As we will show,
discounted sum constraints guarantee safety only in expectation, which is difficult to interpret. 
We propose simple modifications to the penalty term that enable the enforcement of more interpretable constraints, namely, chance constraints \citep{Mesbah2016} and conditional value at risk constraints \citep{Rockafellar2000}, providing bounds on  the probability and the severity of future constraint violations.   
There have been several C-RL algorithms that work with risk sensitivities \citep{Garcia2015a,Chow2018}, but the multiagent context is less studied, and our contributions provide a novel understanding of the safety guarantees provided by C-MARL algorithms.

\begin{wrapfigure}{r}{0.4\textwidth}
  \vspace{-.5cm}
  \begin{center}
    \includegraphics[width=0.38\textwidth]{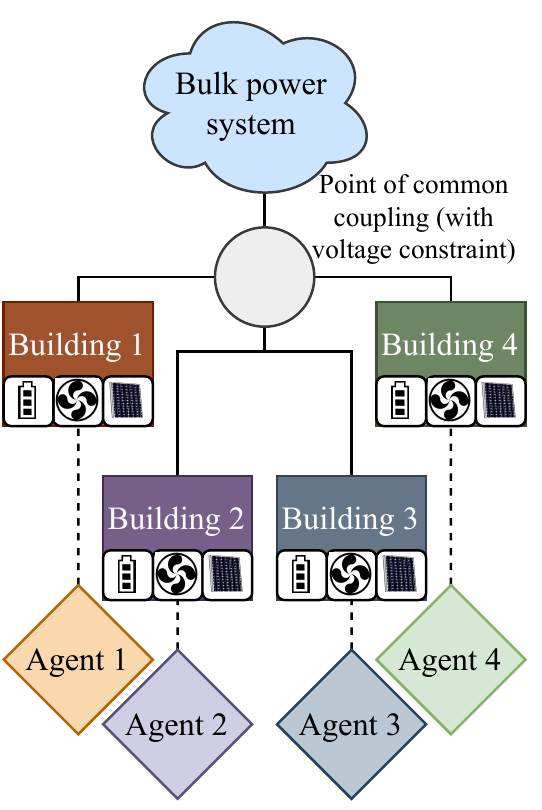}
  \end{center}
  \vspace{-0.5cm}
  \caption{BEM with a voltage constraint at the point of common coupling.
  \vspace{-0.2cm}}
  \label{fig:BEM}
\end{wrapfigure}
The second challenge 
is the fact that the reward is constantly changing as the dual variables are updated, which diminishes the accuracy of value estimates. We quantify this loss of accuracy, and we propose a new value estimation procedure to overcome it. 
Our proposal builds on results in \cite{Tessler2019} showing the affine relationship between the value function and the dual variables. 
We develop a novel class of temporal difference algorithms for value function estimation that directly exploits this observation, giving rise to a value estimate that maintains an accurate derivative with respect to the dual variables. Compared to existing algorithms, our estimates are much more robust to dual variable updates.

The specific C-MARL formulation we study in this paper is inspired by the BEM problem \citep{Molina-Solana2017, Biagioni2022}, illustrated in Figure \ref{fig:BEM}.
The main objective of BEM is to control a building's resources to minimize the cost of energy consumption while maintaining comfort and convenience for the occupants.
However, when BEMs are deployed in multiple buildings, it is critical to ensure that the power network connecting them is safely operated because the uncoordinated control of buildings can cause network-level voltage or power flow violations.
This mandates a level of coordination among agents in the learning stage; thus, we adopt the commonly-studied centralized training/decentralized execution (CTDE) framework \citep{Lowe2017,Foerster2018}, in which a simulator or coordinator provides global state information, constraint evaluations, and Lagrange multipliers (dual variables) to each agent during training. During the testing (execution) phase, we assume that there is no real-time communication between the agents. This stems from the need for privacy and the lack of communication infrastructure in practical systems~\footnote{Even in buildings with advanced metering infrastructure or smart meters, they typically only exchange information with the utility a few times a day.}.

The rest of the paper is organized as follows. In Section \ref{sec:formulation}, we formulate the problem under consideration. In Section \ref{sec:occupation}, we provide an overview of our main interpretive tool, the occupation measure \citep{Borkar1996}. In Section \ref{sec:DRM}, we use the occupation measure to reformulate discounted sum constraints as probabilistic constraints. In Section \ref{sec:value}, we study the value structure of the primal-dual problem and use the results to propose a new value estimation algorithm. In Section \ref{sec:simulations}, we provide some simulation results affirming the contribution of the theoretical observations. 

\subsection{Notation}

The natural numbers and the real numbers are denoted $ \N$ and $\R$, respectively. Given a measurable set $\S$, the set of all possible probability densities over $\S$ is denoted as $\Delta_{\S}$. For any discount factor $\gamma \in (0,1)$ and any sequence $\{y_t\}_{t=0}^T$, the discounted sum operator is $\G_{t=0}^{T}[y_t \mid \gamma] = (1-\gamma) \sum_{t=0}^T \gamma^t y_t$, and $\G_{t=0}^\infty [y_t \mid \gamma] = \lim_{T \rightarrow \infty} \G_{t=0}^T [y_t \mid \gamma]$ if the limit exists. We often drop the second argument $\gamma$ for brevity. The positive component operator is $[y]_+ = \max\{y,0\}$, and the logical indicator function $I[\cdot]$ maps $\{\text{True},\text{False}\}$ to $\{1,0\}$. }{As reinforcement learning (RL) algorithms progress from virtual to cyber-physical applications, it will be necessary to address the challenges of safety and multiple interacting learners. Examples of multi-agent safety-critical systems include power grids \cite{Cui2022}, building energy management systems \cite{Biagioni2022}, autonomous vehicle navigation \cite{Zhou2022a}, and drone swarms \cite{Chen2020}. In each of these applications, agents must learn to operate in a complicated environment, while satisfying various local and system-wide constraints. Such constraints, derived from domain-specific knowledge, are designed to prevent damage to equipment, humans, or infrastructure, or to preclude failure to complete some task or objective. 

One widely-used strategy for solving constrained multi-agent reinforcement learning (C-MARL) problems is the Lagrangian or primal-dual approach \cite{Lu2021a,Li2020,Lee2018,Parnika2021}. In this paradigm, each agent receives a reward signal that is augmented with a penalty term designed to incentivize constraint satisfaction. The magnitude of the penalty term is tuned to incentivize constraint satisfaction while keeping the penalty term from unnecessarily overshadowing the original reward. Although this approach has been shown to converge to a safe joint policy under certain assumptions \cite{Lu2021a}, the constraint that ends up getting enforced is not the same as the one originally expressed. Rather, the Lagrangian approach enforces a bound on the expected value of the discounted sum of constraint evaluations accumulated over a (possibly infinite-horizon) trajectory, known as a \textit{discounted sum constraint} (DSC). 

At any given time step, the safety implications afforded by satisfying a DSC are not immediately obvious. The focus of this paper is to ask the question: \textbf{what is the practical meaning of the algorithmically-convenient discounted sum constraint?} To answer this question, we show that with only slight modifications to the penalty term, \textbf{discounted sum constraints can be reinterpreted as chance constraints or conditional value-at-risk (CVaR) constraints.} The probability density underlying these risk metrics is the \textit{occupation measure} \cite{Borkar1996} which describes the likelihood of visiting a given region of the state space with more weight placed on states that are likely to be visited \textit{earlier} in a trajectory. Although several single-agent RL algorithms deal directly with risk sensitivity \cite{Garcia2015a,Chow2018}, the multi-agent context is less well-studied, and our contribution is to attach new meaning to the safety guarantees provided by existing C-MARL algorithms.

\begin{wrapfigure}{r}{0.4\textwidth}
  \begin{center}
    \includegraphics[width=0.38\textwidth]{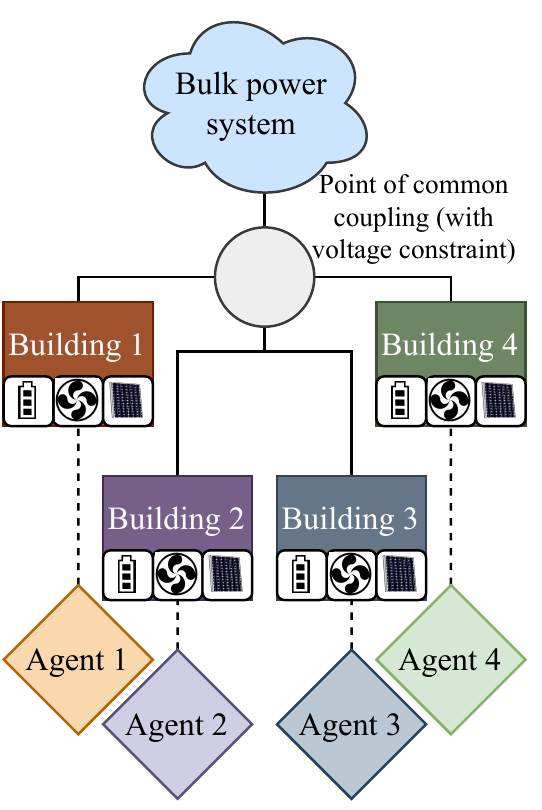}
  \end{center}
  \caption{Building energy management with a voltage constraint at the point of common coupling. Key: PV = photovoltaic array, BESS = battery energy storage system, SWH = smart water heater, SCC = smart climate control system.}
  \label{fig:BEM}
\end{wrapfigure}

The MARL literature includes a wide array of problem formulations and solution techniques. In this paper, we study a specific yet fairly general formulation inspired by the problem of building energy management (BEM) \cite{Biagioni2022}, illustrated in Figure \ref{fig:BEM}. BEM through RL is gaining attention due to the growth of controllable devices such as batteries and smart climate control systems at the grid edge \cite{Molina-Solana2017}. The main objective of BEM is to control a building's resources to minimize the cost of energy consumption, while affording a degree of comfort and convenience to the occupants. Since the thermal dynamics of buildings can be extremely complicated, data-driven techniques pose a promising alternative to traditional model-based control. However, when RL algorithms are deployed on multiple buildings in a locality, the interactions between buildings (through the power system) can be unpredictable, costly, and even dangerous. In this paper, we are particularly interested in whether system-wide constraints will be satisfied when each building is operated by a different agent. One possible constraint is a bound on the voltage at the common node where buildings in a locality connect to the bulk power system. Another example is a constraint on the amount of power or energy extracted from a municipal battery energy storage system. Each of these constraints depends on the real-time power consumption of each building. 

In this paper, we adopt the commonly-studied centralized training/decentralized execution (CTDE) framework \cite{Lowe2017,Foerster2018}. In many MARL situations, global state information must be available to each agent in order to find locally optimal policies, but is not completely necessary for executing the learned policies. The CTDE framework applies to BEM problems in which privacy requirements eliminate the possibility of coordination outside of training time. During training, however, we assume that a simulator or coordinator can provide global state information, constraint evaluations, and Lagrange multipliers (dual variables) to each agent. Although our motivation comes from BEM, we stress that the forthcoming analysis applies to any system that fits into the CTDE framework.

The bulk of the paper is divided into two parts. The first part is devoted to developing a rigorous understanding of our main interpretive tool, the occupation measure, in the context of C-MARL. We show that the occupation measure lies ``between'' the initial and steady-state state visitation probability densities, thus characterizing the near-term behavior of the system under a given joint policy. We then attach a precise meaning to the phrase ``near-term'' using a concept known as the \textit{effective horizon}. The second part of the paper uses the occupation measure to reinterpret DSCs as bounds on \textit{discounted risk metrics}. We show that by making simple modifications to the penalty term, a DSC can be used to control either the probability of incurring a constraint violation (a chance constraint \cite{Mesbah2016}) or the average severity of the worst $(1-\beta)$ fraction of possible constraint evaluations (a CVaR constraint \cite{Rockafellar2000}), in the near term. The analysis is quite general and applies equally to continuous and discrete state and action spaces. We conclude with some simulations affirming the theoretical results and demonstrating the effectiveness of the proposed modifications to the penalty function. 






\subsection{Notation}

The natural numbers $1,2,\ldots$ are denoted $\N$ and the nonnegative reals are denoted $\R_+$. Given a set $\S$, the set of probability densities over $\S$ is denoted $\Delta_{\S}$. For any discount factor $\gamma \in (0,1)$ and any sequence $\{y_t\}_{t=0}^T$, the discounted sum operator is $\G_{t=0}^{T}[y_t \mid \gamma] = (1-\gamma) \sum_{t=0}^T \gamma^t y_t$, and $\G_{t=0}^\infty [y_t \mid \gamma] = \lim_{T \rightarrow \infty} \G_{t=0}^T [y_t \mid \gamma]$ if the limit exists. We will often drop the second argument $\gamma$ for brevity. The positive component operator is $[y]_+ = \max\{y,0\}$ and the logical indicator function $I[\cdot]$ maps $\{\text{True},\text{False}\}$ to $\{1,0\}$. }

\section{Problem formulation} \label{sec:formulation}
\subsection{Constrained MARL}

We consider a noncooperative setting in which $n$ agents pursue individual objectives while subject to global constraints (e.g., a shared resource constraint). We assume there is no real-time communication, and that each agent's action is based only on its local observations. However, policy updates can use global information under the CTDE framework \citep{Lowe2017,Foerster2018}. In this paper, we consider the case of continuous state and action spaces. 

The setting is described by the tuple $(\{\X_i\}_{i \in \n}, \{\U_i\}_{i \in \n}, \{R_i\}_{i \in \n}, f, C, p_0, \gamma)$, where $\n$ is the index set of agents, $\X_i \subset \R^{n_x^i}$ and $\U_i \subset \R^{n_u^i}$ are the state and action spaces of agent $i$, 
and $R_i: \X_i \times \U_i \rightarrow \R$ is the reward function of agent $i$. We assume that the sets $\X_i$ and $\U_i$ are compact for all $i$.  Let $\X = \prod_{i \in \n} \X_i$ and $\U = \prod_{i \in \n} \U_i$ be the joint state and action spaces of the system, respectively. Then $f: \X \times \U \rightarrow \Delta_{\X}$ describes the state transition probabilities, i.e., $f(\cdot \mid x,u)$ is a probability density function. The function $C: \X \rightarrow \R^m$ is used to describe a set of safe states, $\S = \{x \in \X \mid C(x) \leq 0\}.$  

Let $p_0 \in \Delta_{\X}$ denote the initial state probability density and $\gamma \in (0,1)$ be a discount factor. At time $t$, the state, action, and reward of agent $i$ are $x^i_t$, $u^i_t,$ and $r^i_t$, respectively, and constraint $j$ evaluates to $c_t^j = C^j(x_t)$. Using a quantity without a superscript to represent a stacked vector ranging over all $i \in \n$ or all $j \in \{1,\ldots,m\}$, a system trajectory is denoted $\tau = \{(x_t, u_t, r_t, c_t)\}_{t=0}^\infty$.

In the noncooperative C-MARL framework, each agent seeks to learn a policy $\pi_i: \X_i \rightarrow \Delta_{\U_i}$ that maximizes the expected discounted accumulation of individual rewards. We let $\pi: \X \rightarrow \Delta_{\U}$ denote the joint policy, and $f^\pi: \X \rightarrow \Delta_{\X}$ is the state transition probability induced by a joint policy $\pi$.
The tuple $(p_0,f,\pi)$ induces a state visitation probability density at each time step, $p_t^\pi(x) = \int_{\X^t} p_0(x_0) \cdot \prod_{k=1}^t f^\pi(x_k \mid x_{k-1}) \ dx_0\ \ldots\ dx_{k-1}$, and we say $p_\infty^\pi(x) = \lim_{t \rightarrow \infty}p_t^\pi(x)$ for each $x \in \X$ if the limit exists. The collection of visitation probabilities $\{p_t^\pi\}_{t=0}^\infty$ gives rise to a probability density of trajectories $\tau$, denoted $\M \in \Delta_{\prod_{t=0}^\infty(\X \times \A \times \R^n \times \R^m)}$; thus, the objective of each agent can be stated precisely as maximizing $\E_{\tau \sim \M}[\G_{t=0}^\infty r_t^i]$. 

The agents, however, must settle on a joint policy that keeps the system in the safe set $\S$. Due to the stochastic nature of the system, satisfying this constraint at all times is too difficult and in some cases too conservative. A common relaxation procedure is to formulate an augmented reward $\tilde{r}_t^i = r_t^i - \lambda^Tc_t$ where $\lambda \in \R^m_+$, the \emph{Lagrange multiplier} or \emph{dual variable}, is adjusted to incentivize constraint satisfaction. This leads to the primal-dual algorithm for C-MARL, discussed in the next section. The following mild assumption facilitates the analysis. 

\begin{assumption} \label{assumption:bounded} $R^i$, $C^j$, and $p_t^\pi$ are bounded on $\X$ for all $i \in \n$, all $j \in \{1,\ldots,m\}$, and all $t \in \N$. 
\end{assumption}

The boundedness of $R^i$ and $C^j$ is a common assumption \citep{Lu2021a,Tessler2019,Paternain2019} that we will use to exchange the order of limits, sums, and integrals using the dominated convergence theorem. The assumption of bounded $p_t^\pi$ is not strictly necessary and does not change the results; however, we use it throughout the paper to simplify calculations.

\subsection{Primal-dual algorithms}

The augmented reward function leads to the following min-max optimization problem for agent $i$: \begin{align}
    &\min_{\lambda \geq 0} \max_{\pi_i} \E_{\tau \sim \M}\big[ \G_{t=0}^\infty[r_t^i - \lambda^T c_t]\big] \\
    =&\min_{\lambda \geq 0} \max_{\pi_i} \bigg(\E_{\tau \sim \M}\big[ \G_{t=0}^\infty[r_t^i]\big] - \lambda^T \E_{\tau \sim \M}\big[\G_{t=0}^\infty[c_t]\big]\bigg) \label{eqn:11-9-22-1}
\end{align} where \eqref{eqn:11-9-22-1} uses absolute convergence (stemming from Assumption \ref{assumption:bounded}) to rearrange the terms of the infinite sum. Note that the minimization over $\lambda$ is coupled across agents. Any fixed point of \eqref{eqn:11-9-22-1} will satisfy $\E_{\tau \sim \M}[\G_{t=0}^\infty c_t] \leq 0$ because 
if $\E_{\tau \sim \M}[\G_{t=0}^\infty c_t^j] \neq 0$, then the objective value can be reduced by increasing or decreasing $\lambda_j$, unless $\E_{\tau \sim \M}[\G_{t=0}^\infty c_t^j] < 0$ and $\lambda_j = 0$. In other words, the primal-dual method enforces a \emph{discounted sum constraint} derived from the safe set $\S$. 
Although discounted sum constraints are convenient, it is not obvious what they imply about safety guarantees with respect to the original constraints.
We begin our investigation of discounted sum constraints by taking a closer look at a state visitation probability density known as the occupation measure. 

\section{Occupation measure} \label{sec:occupation}
The \textit{occupation measure} describes the average behavior of a Markov process in some sense which will be made precise shortly. As we will show, the occupation measure is instrumental in clarifying the role of discounted sum constraints. In this paper, we use a definition common for continuous-state, infinite-horizon discounted MDPs \citep{Paternain2019,Silver2014}. 

\begin{definition}[Occupation measure] \label{def:11-30-22-1}
The occupation measure $\mu^\pi_\gamma \in \Delta_{\X}$ associated with discount factor $\gamma$, induced by a joint policy $\pi$, is defined for any $x \in \X$ as $\mu^\pi_\gamma(x) = \G_{t=0}^\infty p_t^\pi(x).$
\end{definition}

In this section, we provide some interpretations for the occupation measure before using it to ascribe meaning to discounted sum constraints.
The first question one might ask is whether $\mu^\pi_\gamma$ is itself a pdf. It is, of course, nonnegative, and the following proposition shows it integrates to unity under mild conditions. 

\begin{proposition} \label{thm:10-17-22-1}
Under Assumption \ref{assumption:bounded}, $\int_{\X} \mu^\pi_\gamma(x) dx = 1$. 
\end{proposition}

The proof for Proposition \ref{thm:10-17-22-1} is in Appendix \includeappendix{\ref{appendix:thm}}{A}\footnote{The full paper with appendix is available at \url{https://arxiv.org/abs/2211.16069}.}. What does $\mu^\pi_\gamma$ tell us about the behavior of a system under a given policy? It describes the probability of visiting a certain state but with more weight placed on states that are likely to be visited \textit{earlier} in time. In fact, $\mu^\pi_\gamma$ describes the near-term behavior in the following sense.

\begin{proposition} \label{thm:10-17-22-2} Under Assumption \ref{assumption:bounded}, for any $x \in \X$, the following statements hold:
\begin{enumerate}
    \item $\lim_{\gamma \rightarrow 0^+} \mu^\pi_\gamma(x) = p_0(x).$
    \item $\lim_{\gamma \rightarrow 1^-} \mu^\pi_\gamma(x) = \lim_{t \rightarrow \infty} p_t^\pi$ if the latter limit exists.
\end{enumerate}
\end{proposition}

The proof for Proposition \ref{thm:10-17-22-2} is in Appendix \includeappendix{\ref{appendix:thm}}{A}. Figure \ref{fig:DSD} provides an illustration of the result in Proposition \ref{thm:10-17-22-2} when $p_t^\pi$ evolves as a normal distribution with mean $0.95^t$ and constant variance. The point at which $\mu^\pi_\gamma$ equally resembles $p_0$ and $p_\infty^\pi$ is exactly at $\gamma = 0.95$. 


\begin{wrapfigure}{r}{0.4\textwidth}
  \vspace{-.5cm}
  \begin{center}
    \includegraphics[width=0.38\textwidth]{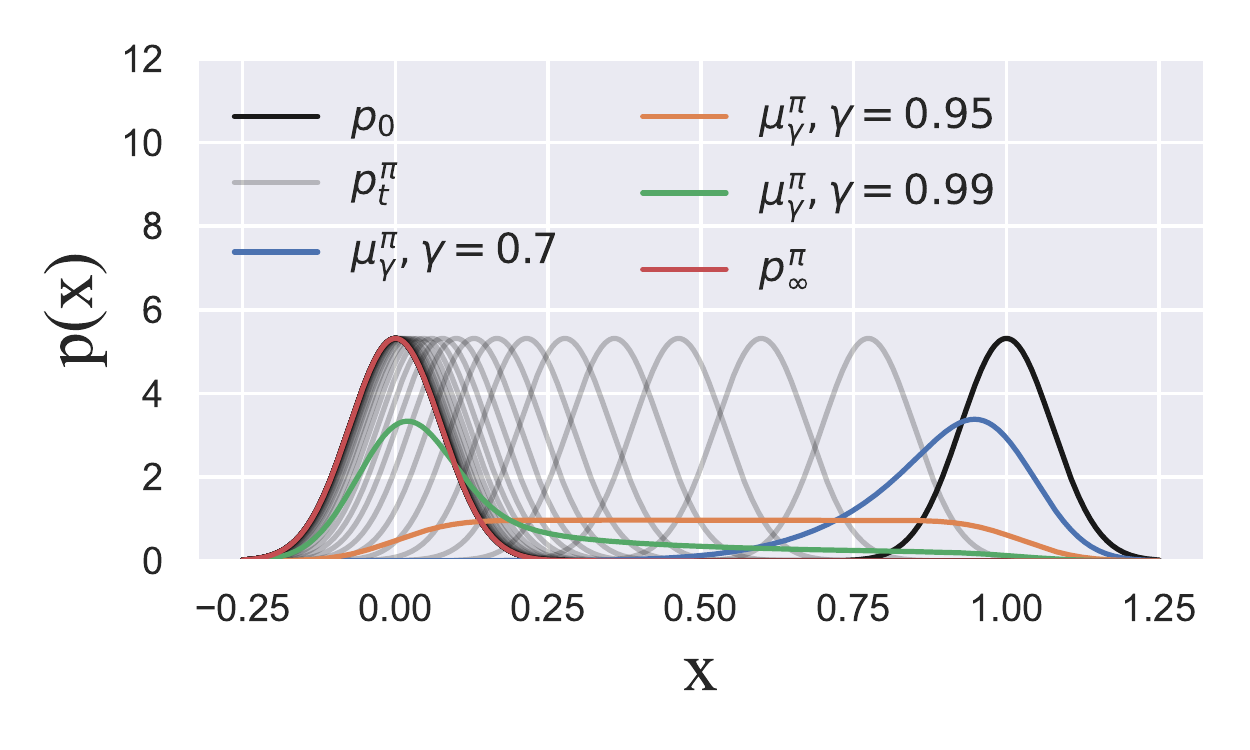}
  \end{center}
  \vspace{-.5cm}
  \caption{Example of the occupation measure for various levels of $\gamma$.}
  \vspace{-.2cm}
  \label{fig:DSD}
\end{wrapfigure}

According to Proposition \ref{thm:10-17-22-2}, the occupation measure describes a state distribution that lies between the initial and long-term behavior of the system. But where exactly does it lie in between these two extremes?
The \textit{effective horizon} of a discounted planning problem is often set to $T_1(\gamma) = \frac{1}{1-\gamma}$, which is the expected termination time if the probability of an episode terminating at any given time step is $(1-\gamma)$ \citep{Paternain2022}; however, the concept of a random stopping time might not be sensible in all applications. Another way to define the effective horizon is to study the geometric accumulation of weights. In this case, the effective horizon can be measured as $T_2(\gamma,\varepsilon) = \min\{K \in \N: \G_{t=0}^{K-1} [1] \geq 1 - \varepsilon\}$, where $\varepsilon \in (0,1)$ is a tolerance. Using either of these two definitions, the occupation measure can be said to describe the behavior of the system from the start time up to the effective horizon. Specifically, one may truncate the sum in Definition \ref{def:11-30-22-1} at the effective horizon to obtain a conceptual understanding of what the occupation measure describes.




Depending on the application, either $T_1$ or $T_2$ can provide a more sensible connection between discounted and finite-horizon problems. But are these two definitions related? The next proposition answers this affirmatively by showing that $T_1$ is actually a special case of $T_2$. 


\begin{proposition} \label{thm:11-9-22-1}
$T_1(\gamma) = T_2(\gamma,\varepsilon)$ when $\varepsilon$ is set to $\gamma^{\frac{1}{1-\gamma}} \approx \frac{1}{e}$.
\end{proposition}

The proof for Proposition \ref{thm:11-9-22-1} is in Appendix \includeappendix{\ref{appendix:thm}}{A}. Proposition \ref{thm:11-9-22-1} is illustrated in Figure \ref{fig:horizon}, where the effective horizon is plotted as a function of $\gamma$ for three different values of $\varepsilon$. With an understanding of the occupation measure as a visitation density describing behavior up to the effective horizon, we can begin to derive meaningful risk-related interpretations of discounted sum constraints. These interpretations lead directly to sensible recommendations for the design of C-MARL algorithms.

\begin{wrapfigure}{r}{0.4\textwidth}
    \vspace{-.2cm}
    \begin{center}
        \includegraphics[width=0.39\textwidth]{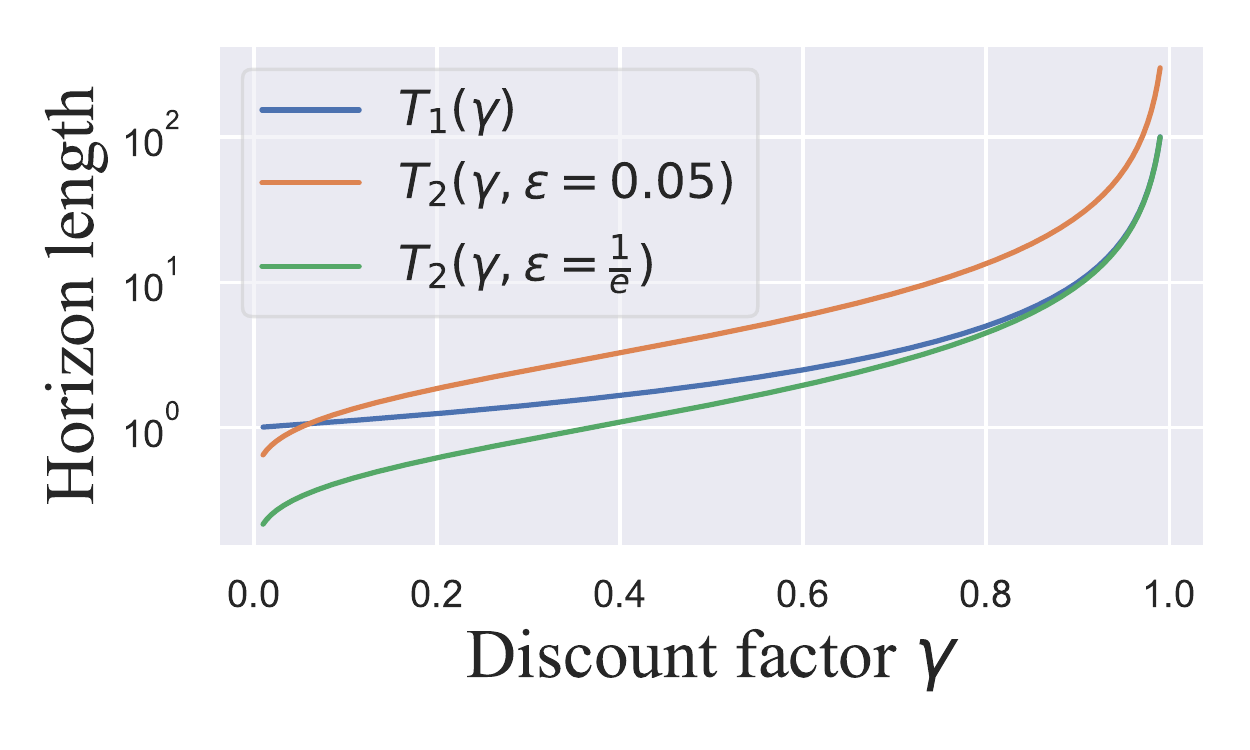}
    \end{center}
    \vspace{-.5cm}
    \caption{Effective horizon length as a function of $\gamma$.}
    \vspace{-.5cm}
    \label{fig:horizon}
\end{wrapfigure}

\section{Discounted risk metrics} \label{sec:DRM}

The discounted sum constraint can naturally be reinterpreted as a certain type of average constraint. In particular, Assumption \ref{assumption:bounded} ensures the equivalence $\E_{\tau \sim \M}[\G_{t=0}^\infty C(x_t)] = \E_{x \sim \mu^\pi_\gamma}[C(x)]$ \citep{Paternain2019}. This near-term average does not relate to any well-known risk metrics and hence does not provide a practical safety guarantee.  
In general, information about the mean of a distribution cannot be used to infer information about its tails; however, simple changes to the penalty function can yield information about either the \textit{probability} of incurring a constraint violation or the expected \textit{severity} of constraint violations. 

\begin{proposition}[Near-term probability of constraint violations] \label{thm:10-18-22-1} Suppose that for some $\delta_j \in [0,1]$ and $\alpha_j \in \R$\ , we have $\E_{\tau \sim \M}[\G_{t=0}^\infty I[C^j(x_t) \geq \alpha_j]] \leq \delta_j$.
Then under Assumption \ref{assumption:bounded}, $\textup{Pr}\{C^j(x) \geq \alpha_j \mid x \sim \mu^\pi_\gamma\} \leq \delta_j.$
\end{proposition}

\begin{proof}
$\E_{\tau \sim \M}[\G_{t=0}^\infty I[C^j(x_t) \geq \alpha_j]] = \E_{x \sim \mu^\pi_\gamma}[I[C^j(x) \geq \alpha_j]] = \textup{Pr}\{C^j(x) \geq \alpha_j \mid x \sim \mu^\pi_\gamma\}$. The first equality uses Assumption \ref{assumption:bounded} to apply an equivalence established in e.g. \cite{Paternain2019}.
The second equality follows from the definition of expectation.
\end{proof}


Proposition \ref{thm:10-18-22-1} makes it easy to enforce chance constraints using primal-dual methods. When the penalty term $C^j(x)$ is replaced by the quantity $I[C^j(x) \geq \alpha_j] - \delta_j$, the primal-dual algorithm enforces $\E_{\tau \sim \M}[\G_{t=0}^\infty I[C^j(x_t) \geq \alpha_j]] - \delta_j \leq 0$. By Proposition \ref{thm:10-18-22-1}, this guarantees that $\textup{Pr}\{C^j(x) \geq \alpha_j \mid x \sim \mu^\pi_\gamma\} \leq \delta_j$. Because the probability of constraint violations is defined with $x$ varying over $\mu^\pi_\gamma$, we call the resulting guarantee a \textit{near-term} or \textit{discounted chance constraint.} This can be repeated for each $j \in \{1,\ldots,m\}$, providing a set of bounds on the probability of violating \emph{each} constraint by more than its tolerance $\alpha_j$. On the other hand, we can control the probability of violating \textit{any} constraint as follows. Define the statement $C(x) \geq \alpha$ to be true if $C^j(x) \geq \alpha_j \ \forall\ j \in \{1,\ldots,m\},$ and false otherwise. Then, applying Proposition \ref{thm:10-18-22-1} to the test condition $C(x) \geq \alpha$ will result in a bound on $\textup{Pr}\{C(x) \geq \alpha \mid x \sim \mu^\pi_\gamma\}$.

While discounted chance constraints enable one to control the \textit{probability} of extreme events in the near future, conditional value at risk constraints \citep{Rockafellar2000} afford control over the \textit{severity} of such events.

\begin{definition}[\cite{Rockafellar2000}]
Given a risk level $\beta \in (0,1)$, a cost $h: \X \rightarrow \R$, and a probability density $\mu$ on $\X$, the value at risk (VaR) and conditional value at risk (CVaR) are defined as: \begin{gather*}
    \textup{VaR}(\beta, h, \mu) = \min \{ \alpha \in \R: \textup{Pr}\{h(x) \leq \alpha \mid x \sim \mu\} \geq \beta\}, \\
    \textup{CVaR}(\beta,h,\mu) = \frac{1}{1-\beta} \int_{h(x) \geq \textup{VaR}(\beta,h,\mu)} h(x) \mu(x) dx.
\end{gather*}
\end{definition}

In other words, $\textup{VaR}(\beta, h, \mu)$ is the least upper bound on $h$ that can be satisfied with probability $\beta$, while $\textup{CVaR}(\beta, h, \mu)$ describes the expected value in the VaR-tail of the distribution of $h$. CVaR characterizes the expected severity of extreme events, which can be defined precisely as the $(1-\beta)$ fraction of events $x$ with the worst outcomes as ranked by the cost incurred, $h(x)$. The VaR and CVaR for $h(x) = x,$ when $x$ follows a standard normal distribution, are illustrated in Figure \ref{fig:var_cvar}, where the shaded region has an area of $(1-\beta)$. For the rest of the paper, we assume that the cdf of $h(x)$ is continuous when $x \sim \mu$. For further details and for cases in which this assumption does not hold, we refer the reader to \cite{Rockafellar2002}. 




\begin{proposition}[Near-term CVaR] \label{thm:10-20-22-1}
For any $\alpha_j \geq 0$, suppose that $\E_{\tau \sim \M}[\G_{t=0}^\infty [[C^j(x_t) - \alpha_j]_+]] \leq \delta_j$. Then, $\textup{CVaR}(\beta,C^j,\mu^\pi_\gamma) \leq \alpha_j + (1-\beta)^{-1} \delta_j.$
\end{proposition}

\begin{proof}
Under Assumption \ref{assumption:bounded}, the identity $\E_{\tau \sim \M}[\G_{t=0}^\infty [C^j(x_t) - \alpha_j]_+] = \E_{x \sim \mu^\pi_\gamma}[[C^j(x) - \alpha_j]_+]$ holds \citep{Paternain2019}. Next, we use the fact that the CVaR is the minimum value of the convex function in $\alpha_j$ given by $F(\alpha_j \mid \beta,C^j,\mu^\pi_\gamma) := \alpha_j + (1-\beta)^{-1} \E_{x \sim \mu^\pi_\gamma}[[C^j(x)-\alpha_j]_+]$ \citep{Rockafellar2000}; thus, $F$ provides an upper bound on CVaR. Some rearranging leads to the result.
\end{proof}

\begin{wrapfigure}{r}{0.4\textwidth}
    \vspace{-1cm}
    \begin{center}
    \includegraphics[width= 0.4\textwidth]{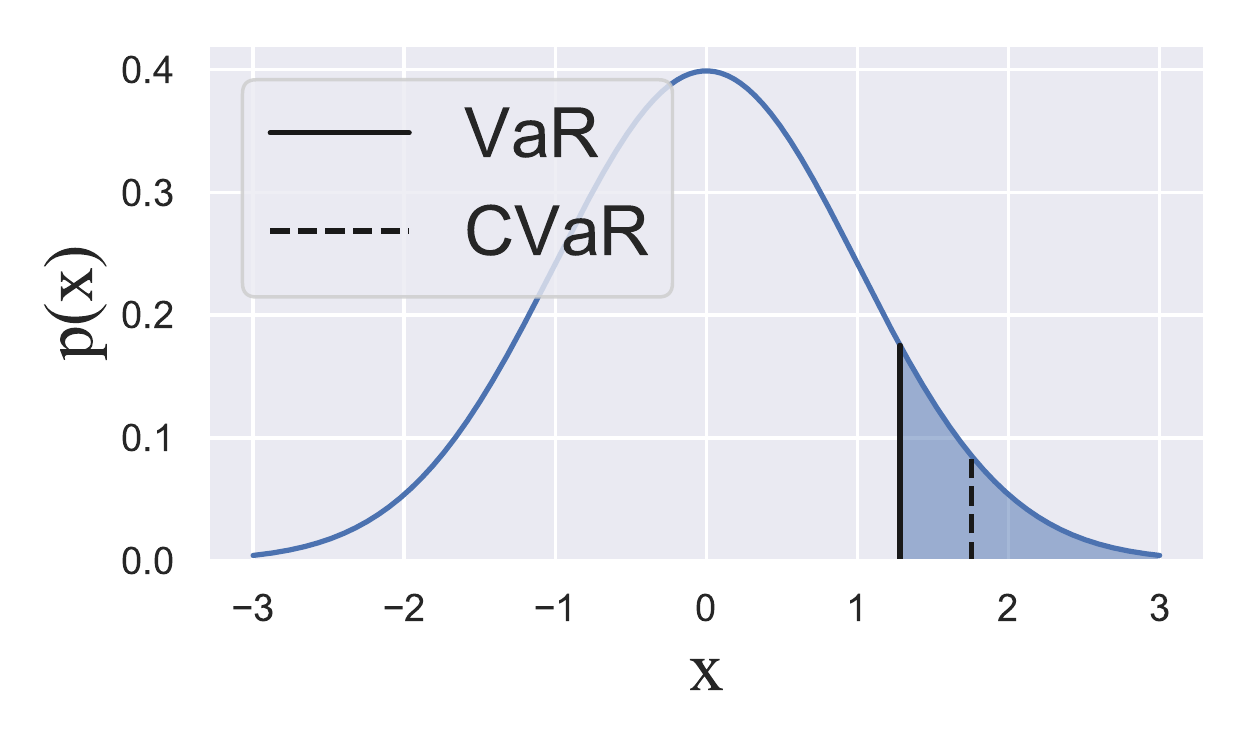}
    \end{center}
    \vspace{-0.5cm}
    \caption{Example of VaR and CVaR at risk level $\beta = 0.9$.}
    \vspace{-0.2cm}
    \label{fig:var_cvar}
\end{wrapfigure}

Similar to the chance-constrained case, Proposition \ref{thm:10-20-22-1} makes it easy to enforce CVaR constraints in the primal-dual algorithm. Here, the penalty term used is $[C^j(x) - \alpha_j]_+ - \delta_j$. Using this penalty, the algorithm enforces $\E_{\tau \sim \M}[\G_{t=0}^\infty [[C^j(x_t) - \alpha_j]_+]] - \delta_j \leq 0$, which by Proposition \ref{thm:10-20-22-1} implies $\textup{CVaR}(\beta,C^j,\mu^\pi_\gamma) \leq \alpha_j + (1-\beta)^{-1} \delta_j.$ By repeating for each $j \in \{1,\ldots,m\}$, we can bound the expected severity of the constraint violations for each of the $m$ constraints. Because the CVaR constraint is defined with $x$ varying over $\mu^\pi_\gamma$, the resulting guarantee is called a \textit{near-term} or \textit{discounted CVaR constraint.}

To obtain a tight bound on the CVaR, $\alpha_j$ must be set to $\textup{VaR}(\beta,C^j,\mu^\pi_\gamma)$, which minimizes the function $F$ introduced in the proof of Proposition \ref{thm:10-20-22-1} \citep{Rockafellar2000}. Unfortunately, the VaR is not known ahead of time. \cite{Chow2018} include $\alpha_j$ as an optimization variable in the learning procedure, but extending their technique to the multiagent setting is not straightforward. Our approach is to include it as a tunable hyperparameter. Simulation results in Section \ref{sec:simulations} show that it is easy to choose $\alpha_j$ to give a nearly tight bound.


\newversion{\section{Primal-dual value functions} \label{sec:value}
In this section, we investigate challenges with value estimation in the primal-dual regime. The fact that the reward to each agent is constantly changing (due to dual variable updates) makes it difficult to accurately estimate state values. To quantify this decrease in accuracy, we introduce the value functions induced by the joint policy $\pi$,
$\{V_\pi^i: \X \times \R \rightarrow \R\}_{i \in \n},\ \{V_{R,\pi}^i: \X \rightarrow \R\}_{i \in \n},\ V_{C,\pi}: \X \rightarrow \R^m$ where: 
 \begin{gather}
    V^i_\pi(x,\lambda) = \E_{\tau \sim \M} [\G_{t=0}^\infty r_t^i - \lambda^Tc_t \mid x_0 = x], \label{eqn:11-20-22-1}\\ V^i_{R,\pi}(x) = \E_{\tau \sim \M} [\G_{t=0}^\infty r_t^i \mid x_0 = x],\quad
    V_{C,\pi}(x) = \E_{\tau \sim \M} [\G_{t=0}^\infty c_t \mid x_0 = x].
\end{gather}Note that $c_t$ could be modified as indicated in Section \ref{sec:DRM}, and the following results would hold for the modified penalty function.

Obviously, it is impossible to learn an accurate value function when $\lambda$ is unknown and changing; however, simply making $\lambda$ available to a value function approximator does not guarantee good generalization beyond previously seen values of $\lambda$. Having a good estimate of the \emph{derivative} of the value function with respect to $\lambda$ will ensure accuracy under small perturbations to the dual variables. Fortunately, this derivative is easy to obtain. Formally, under Assumption \ref{assumption:bounded}, we can write $V_\pi^i(x,\lambda) = V_{R,\pi}^i(x) - \lambda^T V_{C,\pi}(x)$ \citep{Tessler2019}, and therefore, $\nabla_\lambda V_\pi^i(x,\lambda) = -V_{C,\pi}(x)$. By learning $V_{R,\pi}^i$ and $V_{C,\pi}$ as separate functions and then combining them using the true value of $\lambda$, we can construct a value estimate whose derivative with respect to the dual variables is as accurate as our estimate of $V_{C,\pi}$ itself. This estimate will be more robust to small changes in $\lambda$. We will refer to this type of value estimate as a \emph{structured value function} or a \emph{structured critic}. 

\begin{proposition} \label{thm:11-14-22-1}
Let $\bar{c} = \E_{x \sim \mu^\pi_\gamma}[C(x)]$ and $\Sigma_C^2 = \E_{x \sim \mu^\pi_\gamma}[(\bar{c} - C(x))(\bar{c} - C(x))^T]$. Suppose $\lambda$ is randomly varying with mean $\bar{\lambda}$ and variance $\Sigma_\lambda^2$. Using a structured value function approximator can reduce the mean square temporal difference error by up to $\textup{Tr}[\Sigma_\lambda^2 \cdot (\Sigma_C^2 + \bar{c}\bar{c}^T)]$. 
\end{proposition}


\begin{wrapfigure}{r}{0.4\textwidth}
    \vspace{-.5cm}
    \begin{center}
    \includegraphics[width = 0.4\textwidth]{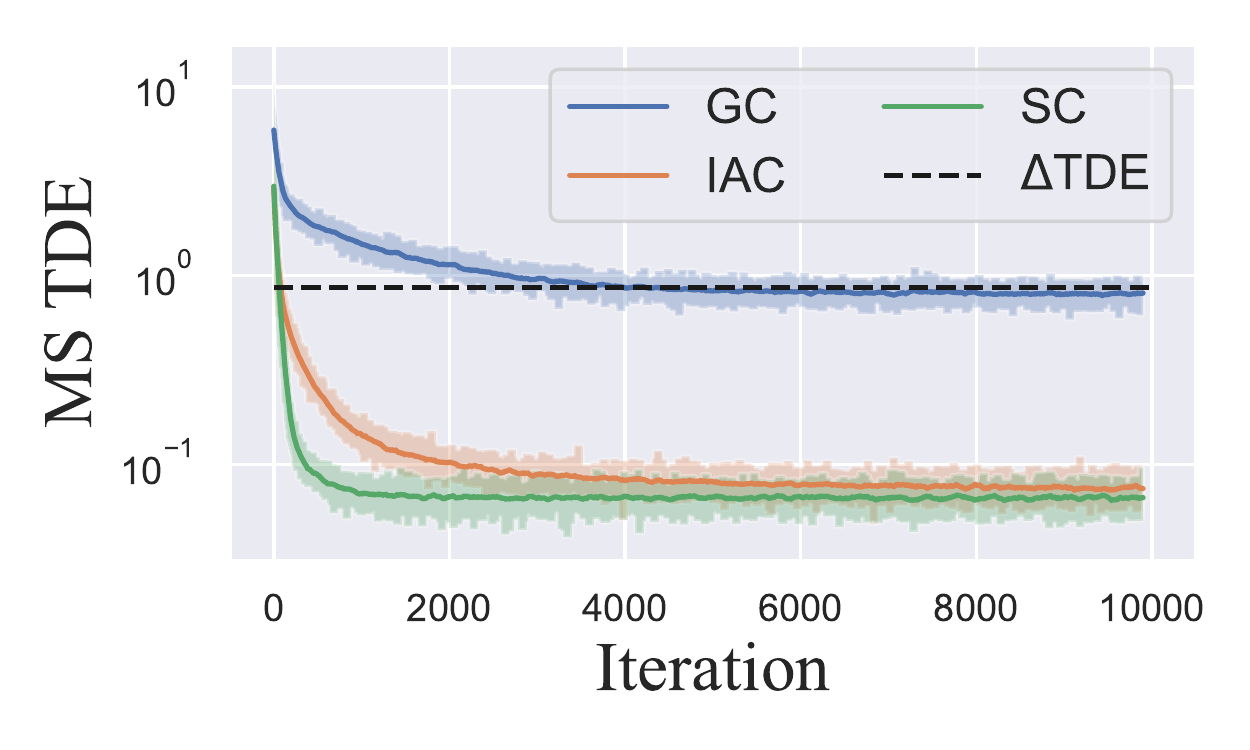}
    \end{center}
    \vspace{-.5cm}
    \caption{Temporal difference error trajectories in a simple policy evaluation task.}
    \vspace{-.2cm}
    \label{fig:CCA}
\end{wrapfigure}

The proof of Proposition \ref{thm:11-14-22-1} is in Appendix \includeappendix{\ref{appendix:thm}}{A}. Figure \ref{fig:CCA} illustrates Proposition \ref{thm:11-14-22-1} in a simple value estimation task with quadratic rewards, linear dynamics and policies, linear state constraints, and randomly varying $\lambda$. The \textit{generic critic} (GC) is a value function modeled as a quadratic function of the state only. The \textit{input-augmented critic} (IAC) is a value function modeled as an unknown quadratic function of the state and dual variables, while the \textit{structured critic} (SC) is modeled using $\hat{V}_{\pi}^i = \hat{V}_{R,\pi}^i - \lambda^T\hat{V}_{C,\pi}$ with quadratic $\hat{V}_{R,\pi}^i$ and linear $\hat{V}_{C,\pi}$ trained on their respective signals. 

The dashed line in Figure \ref{fig:CCA} is at the value $\textup{Tr}[\Sigma_\lambda^2 \cdot (\Sigma_C^2 + \bar{c}\bar{c}^T)]$ predicted in Proposition \ref{thm:11-14-22-1}. 
In this simple value estimation task, high accuracy can be achieved when conditioning on the randomly varying $\lambda$; however, having an accurate estimate of $\nabla_\lambda V_{\pi}^i$ by using a structured critic is also shown to help.
Although in practice $\bar{\lambda}$ and $\Sigma_{\lambda}^2$ change over time, the simulation results in Section \ref{sec:simulations} confirm that using structured critics improves performance. 
The loss function for value function approximation is therefore given by: \begin{align}
    TDE(x,x') = [R^i(x^i) + \gamma \hat{V}_{R,\pi}^i((x^i)') - \hat{V}_{R,\pi}^i(x^i)]^2 + \|C(x) + \gamma \hat{V}_{C,\pi}(x') - \hat{V}_{C,\pi}(x)\|_2^2 \label{eqn:11-21-22-1}
\end{align} 
where $x \in \X$ and $x' \sim f^\pi(x)$. Equation \eqref{eqn:11-21-22-1} is simply a sum of squared temporal difference errors over the set of $m+1$ value functions. For algorithmic details, we refer the reader to Appendix \includeappendix{\ref{appendix:simulations}}{B}. 

}{}

\section{Simulations} \label{sec:simulations}
In our simulations, we sought to demonstrate the effectiveness of the penalty modifications and structured critic proposed in sections \ref{sec:DRM} and \ref{sec:value}. We tested our findings in a modified multiagent particle environment\footnote{Code for the environments is available at \url{github.com/dtabas/multiagent-particle-envs}.} \citep{Lowe2017} with two agents pursuing individual objectives subject to a constraint on the joint state. The state of each agent is its position and velocity in $\R^2$, i.e. $x^i = \m{y^{iT} & v^{iT}}^T$ where $y^i \in \R^2$ is the position and $v^i \in \R^2$ is the velocity of agent $i$. The objective of each agent is to drive its position $y^i$ to a landmark $y^{i*} \in \R^2$, while making sure that the agent ensemble satisfies the safety constraint. The reward and constraint functions are given by:
\begin{align}
    R^i(y^i) = -\xi_i \|y^i - y^{i*}\|_2^2,\quad C(y) = \textbf{1}^T y
\end{align} where $\xi_i > 0$ is a constant and $y = \m{y^{1T} & y^{2T}}^T$ is the position of the agent ensemble.

\begin{wrapfigure}{r}{0.4\textwidth}
    \vspace{-.5cm}
  \begin{center}
    \includegraphics[width=0.38\textwidth]{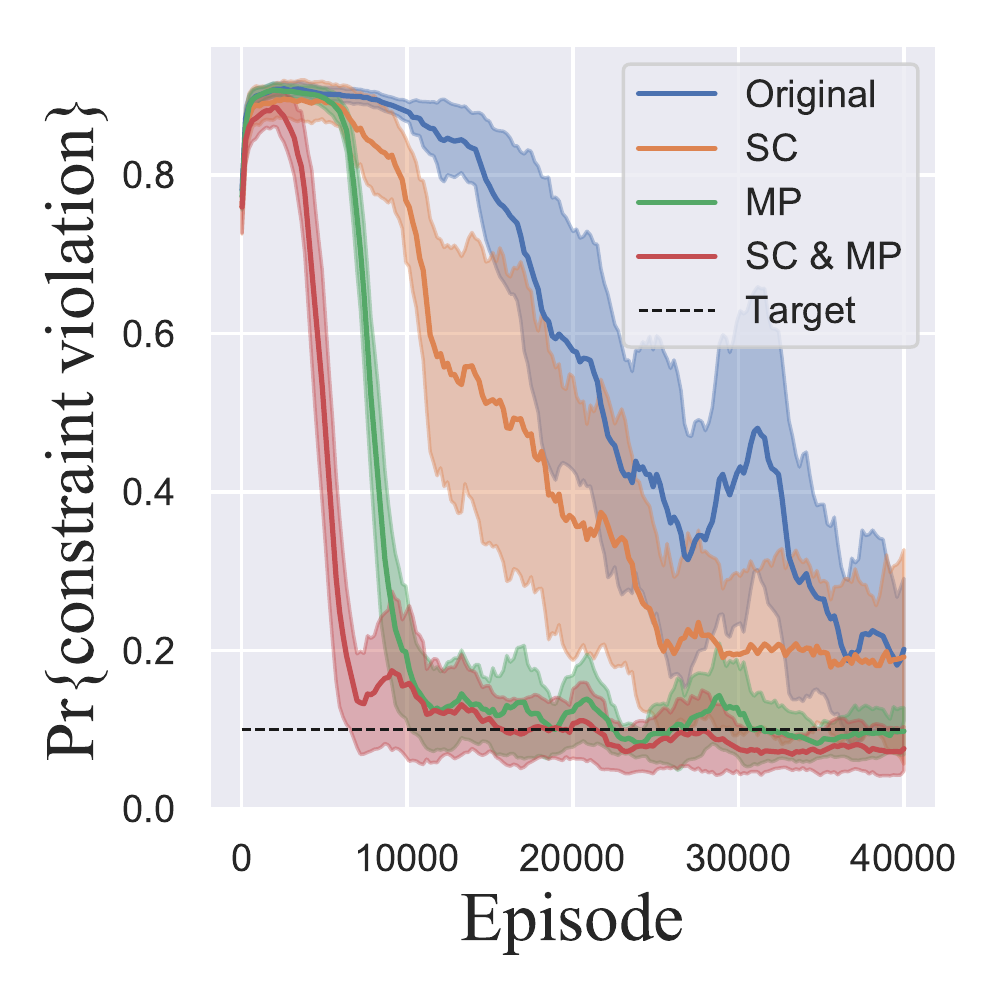}
  \end{center}
  \vspace{-.5cm}
  \caption{$\textup{Pr}\{C(x) \geq 0.1 \mid x \sim \mu^\pi_\gamma\}$ measured throughout training. Key: SC = structured critic, MP = modified penalty (Prop. \ref{thm:10-18-22-1}). Both modifications speed convergence to a safe policy. The shaded region represents $\pm 1$ standard deviation across 5 training runs.}
  \vspace{-.2cm}
  \label{fig:sim_chance}
\end{wrapfigure}

The landmark $y^* = \m{y^{1*T} & y^{2*T}}^T$ is stationed outside of the safe region $\S = \{y \mid C(y) \leq 0\}.$ Thus, the agents cannot both reach their goals while satisfying $C(y) \leq 0$. 
To train the agents to interact in this environment, we used a modified version of the EPyMARL codebase\footnote{Code for the algorithms is available at \url{github.com/dtabas/epymarl}.} \citep{Papoudakis2020}. We tested several MARL algorithms, including MADDPG \citep{Lowe2017}, COMA \citep{Foerster2018}, and MAA2C \citep{Papoudakis2020}. We decided to use the MAA2C algorithm because it consistently produced the best results and because as a value function-based algorithm, it provided the most straightforward route to implementing the changes proposed in Section \ref{sec:value}. Details of the algorithm, pseudocode, hyperparameters, and supplementary simulation results are provided in Appendix \includeappendix{\ref{appendix:simulations}}{B}.

For each risk metric described in Section \ref{sec:DRM}, we tested the convergence of the agents to a safe policy with and without modifications to the penalty and value functions. Figure \ref{fig:sim_chance} shows the results when we make the substitution $C(x) \leftarrow I[C(x) \geq \alpha] - \delta$ in the penalty function to enforce a chance constraint, $\textup{Pr}\{C(x) \geq \alpha \mid x \sim \mu^\pi_\gamma\} \leq \delta$ with $\alpha$ and $\delta$ each set to $0.1$. The modified penalty function performs the best as a chance constraint-enforcing signal (red and green lines in Figure~\ref{fig:sim_chance}). Whether or not the penalty function is modified, the structured critic 
finds safer policies throughout training (red vs. green and orange vs. blue lines).

Figure \ref{fig:sim_cvar} shows the results when we make the substitution $C(x) \leftarrow [C(x) - \alpha]_+ - \delta$ in the penalty function to enforce the constraint $\textup{CVaR}(\beta,C,\mu^\pi_\gamma) \leq \alpha + (1-\beta)^{-1} \delta$. Using the modified penalty (red and green lines in Figure \ref{fig:sim_cvar}) drives the CVaR upper bound (drawn in dashed lines) to the target value, and due to the choice of $\alpha$, this bound is nearly tight. On the other hand, using the original penalty results in an overly conservative policy that achieves low risk at the expense of rewards (right panel). We also point out that when using the modified penalty with the structured critic, the CVaR is lower throughout training compared to when the generic critic is used, indicating improved effectiveness in enforcing limits on risk.

We chose $\alpha$ using the following heuristic, to make the bound on CVaR nearly tight. The ``correct'' value of $\alpha$ that would achieve a tight bound is $\textup{VaR}(\beta,C,\mu^\pi_\gamma)$. Moreover, the upper bound that we used is convex and continuously differentiable in $\alpha$ \citep{Rockafellar2000}; therefore, small errors in $\alpha$ will lead to small errors in the upper bound on CVaR, and any approximation of the VaR will suffice. We obtained an approximation simply by running the simulation once with $\alpha$ set to zero and computing $\textup{VaR}(\beta,C,\mu^\pi_\gamma)$ over some test trajectories. If necessary, the process could be repeated additional times. Alternatively, $\alpha$ could be tuned adaptively by computing VaR online, but the stability of such a procedure would need further investigation. 

\begin{figure}
  \begin{center}
    \includegraphics[width=0.8\textwidth]{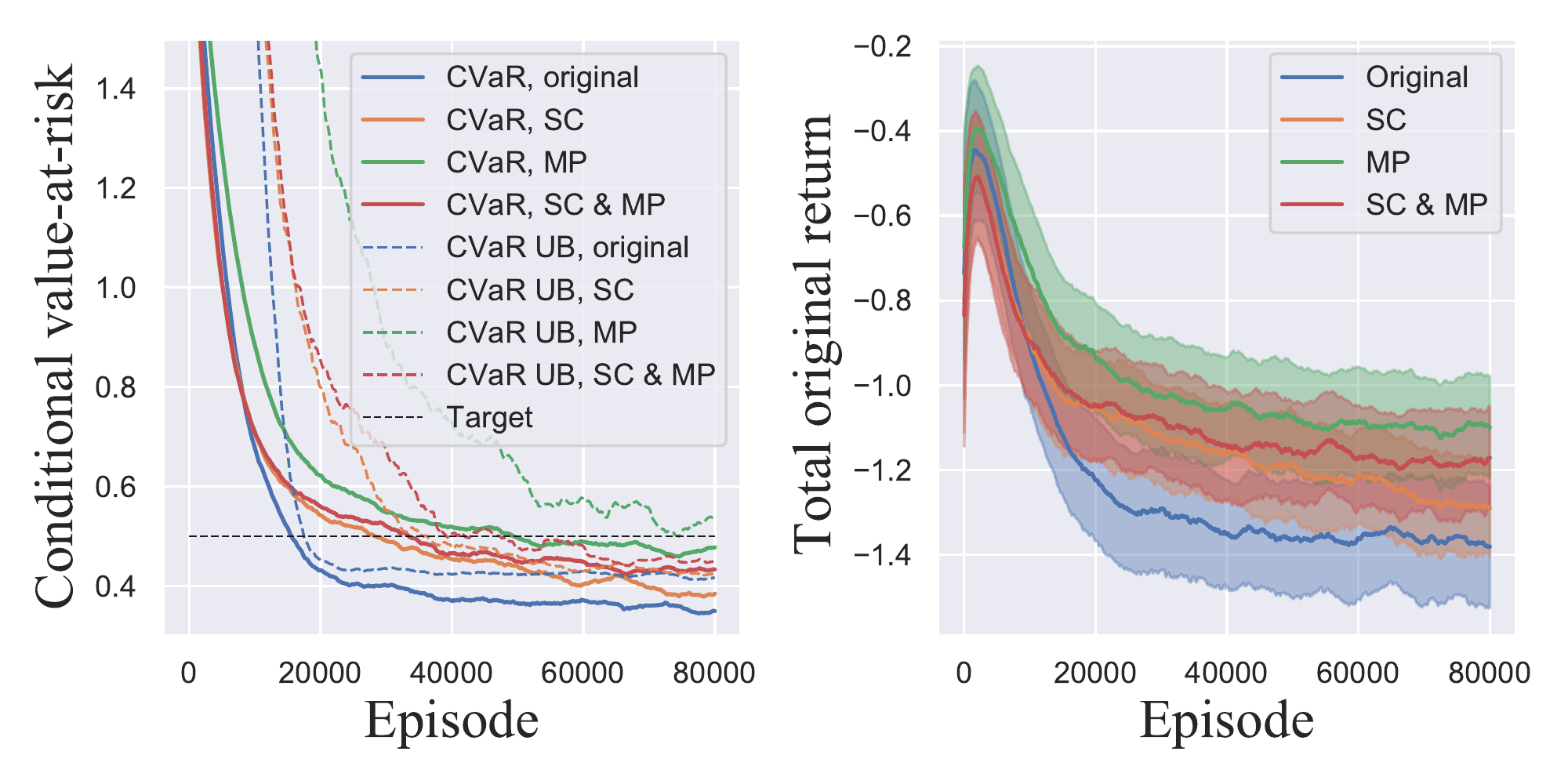}
  \end{center}
  \vspace{-.5cm}
  \caption{$\textup{CVaR}(\beta=0.9,C,\mu^\pi_\gamma)$ measured throughout training. Key: SC = structured critic, MP = modified penalty (Prop. \ref{thm:10-20-22-1}). The dashed lines represent the CVaR upper bound used in Prop. \ref{thm:10-20-22-1}. 
  The panel on the right shows progress toward the original objective through the total original returns, $\sum_{i=1}^2 \G_{t=0}^T r^i_t,$ without penalty terms. The shaded region represents $\pm 1$ standard deviation across 5 training runs. The rewards increase then decrease because the agents first learn to navigate towards the landmark, which is outside the safe region, then learn to back off to satisfy the constraint.}
  \label{fig:sim_cvar}
\end{figure}

\section{Conclusion}
In this paper, we studied the effect of primal-dual algorithms on the structure of C-MARL problems. First, we used the occupation measure to study the effect of the penalty term on safety. We showed that using the constraint function as the penalty enforces safety only in expectation, but by making simple modifications to the penalty term, one may enforce meaningful probabilistic safety guarantees, namely, chance and CVaR constraints. These risk metrics are defined over the occupation measure, leading to notions of safety in the near term. 
Next, we studied the effect of the penalty term on the value function. When the dual variable and constraint evaluation signals are available, it is easy to model the relationship between the penalty term and the value function. By exploiting this structure, the accuracy of the value function can be improved. We demonstrated the usefulness of both of these insights in a constrained multiagent particle environment, showing that convergence to a low-risk policy is accelerated. 
One open question is the effect of primal-dual methods on game outcomes. Some agents might pay a higher price than others for modifying their policies to satisfy system-wide constraints. Understanding and mitigating this phenomenon will be the focus of future work.

 \acks{D. Tabas and A. S. Zamzam would like to thank Patrick Emami and Xiangyu Zhang for helpful discussions in the early stages of this work, and Georgios Papoudakis for advice on software implementation. This work is partially supported by the National Science Foundation Graduate Research Fellowship Program under Grant No. DGE-2140004. Any opinions, findings, conclusions, or recommendations expressed in this material are those of the authors and do not necessarily reflect the views of the National Science Foundation.
This work was authored in part by the National Renewable Energy Laboratory (NREL), operated by Alliance for Sustainable Energy, LLC, for the U.S. Department of Energy (DOE) under Contract No. DE-AC36-08GO28308. The work of A. S. Zamzam was supported by the Laboratory Directed Research and Development (LDRD) Program at NREL. The views expressed in the article do not necessarily represent the views of the DOE or the U.S. Government. The U.S. Government retains and the publisher, by accepting the article for publication, acknowledges that the U.S. Government retains a nonexclusive, paid-up, irrevocable, worldwide license to publish or reproduce the published form of this work, or allow others to do so, for U.S. Government purposes.}

\bibliography{references2}

\begin{thebibliography}{25}
\providecommand{\natexlab}[1]{#1}
\providecommand{\url}[1]{\texttt{#1}}
\expandafter\ifx\csname urlstyle\endcsname\relax
  \providecommand{\doi}[1]{doi: #1}\else
  \providecommand{\doi}{doi: \begingroup \urlstyle{rm}\Url}\fi

\bibitem[Biagioni et~al.(2022)Biagioni, Zhang, Wald, Vaidhynathan, Chintala,
  King, and Zamzam]{Biagioni2022}
David Biagioni, Xiangyu Zhang, Dylan Wald, Deepthi Vaidhynathan, Rohit
  Chintala, Jennifer King, and Ahmed~S. Zamzam.
\newblock {PowerGridworld: A Framework for Multi-Agent Reinforcement Learning
  in Power Systems}.
\newblock \emph{Proceedings of the 2022 13th ACM International Conference on
  Future Energy Systems}, pages 565--570, 2022.

\bibitem[Borkar and Bhatt(1996)]{Borkar1996}
Vivek~S. Borkar and Abhay~G. Bhatt.
\newblock {Occupation Measures for Controlled Markov Processes:
  Characterization and Optimality}.
\newblock \emph{The Annals of Probability}, 24\penalty0 (3):\penalty0
  1531--1562, 1996.

\bibitem[Chen et~al.(2018)Chen, Saulnier, Atanasov, Lee, Kumar, Pappas, and
  Morari]{Chen2018c}
Steven Chen, Kelsey Saulnier, Nikolay Atanasov, Daniel~D. Lee, Vijay Kumar,
  George~J. Pappas, and Manfred Morari.
\newblock {Approximating Explicit Model Predictive Control Using Constrained
  Neural Networks}.
\newblock In \emph{Proceedings of the American Control Conference}, pages
  1520--1527, 2018.

\bibitem[Chen et~al.(2020)Chen, Chang, and Zhang]{Chen2020}
Yu~Jia Chen, Deng~Kai Chang, and Cheng Zhang.
\newblock {Autonomous Tracking Using a Swarm of UAVs: A Constrained Multi-Agent
  Reinforcement Learning Approach}.
\newblock \emph{IEEE Transactions on Vehicular Technology}, 69\penalty0
  (11):\penalty0 13702--13717, 2020.

\bibitem[Chow et~al.(2018)Chow, Ghavamzadeh, Janson, and Pavone]{Chow2018}
Yinlam Chow, Mohammad Ghavamzadeh, Lucas Janson, and Marco Pavone.
\newblock {Risk-constrained reinforcement learning with percentile risk
  criteria}.
\newblock \emph{Journal of Machine Learning Research}, 18:\penalty0 1--51,
  2018.

\bibitem[Cui et~al.(2022)Cui, Li, and Zhang]{Cui2022}
Wenqi Cui, Jiayi Li, and Baosen Zhang.
\newblock {Decentralized safe reinforcement learning for inverter-based voltage
  control}.
\newblock \emph{Electric Power Systems Research}, 211\penalty0 (108609), 2022.

\bibitem[Foerster et~al.(2018)Foerster, Farquhar, Afouras, Nardelli, and
  Whiteson]{Foerster2018}
Jakob~N. Foerster, Gregory Farquhar, Triantafyllos Afouras, Nantas Nardelli,
  and Shimon Whiteson.
\newblock {Counterfactual multi-agent policy gradients}.
\newblock \emph{32nd AAAI Conference on Artificial Intelligence}, pages
  2974--2982, 2018.

\bibitem[Garc{\'{i}}a and Fern{\'{a}}ndez(2015)]{Garcia2015a}
Javier Garc{\'{i}}a and Fernando Fern{\'{a}}ndez.
\newblock {A Comprehensive Survey on Safe Reinforcement Learning}.
\newblock \emph{Journal of Machine Learning Research}, 16:\penalty0 1437--1480,
  2015.

\bibitem[Lee et~al.(2018)Lee, Yoon, and Hovakimyan]{Lee2018}
Donghwan Lee, Hyungjin Yoon, and Naira Hovakimyan.
\newblock {Primal-Dual Algorithm for Distributed Reinforcement Learning:
  Distributed GTD}.
\newblock In \emph{Proceedings of the IEEE Conference on Decision and Control},
  pages 1967--1972, 2018.

\bibitem[Li et~al.(2020)Li, Jin, Wang, Yan, and Zha]{Li2020}
Wenhao Li, Bo~Jin, Xiangfeng Wang, Junchi Yan, and Hongyuan Zha.
\newblock {F2A2: Flexible Fully-decentralized Approximate Actor-critic for
  Cooperative Multi-agent Reinforcement Learning}.
\newblock \emph{arXiv: 2004.11145}, pages 1--42, 2020.

\bibitem[Lowe et~al.(2017)Lowe, Wu, Tamar, Harb, Abbeel, and
  Mordatch]{Lowe2017}
Ryan Lowe, Yi~Wu, Aviv Tamar, Jean Harb, Pieter Abbeel, and Igor Mordatch.
\newblock {Multi-Agent Actor-Critic for Mixed Cooperative-Competitive
  Environments}.
\newblock In \emph{31st Conference on Neural Information Processing Systems},
  2017.

\bibitem[Lu et~al.(2021)Lu, Zhang, Chen, Ba{\c{s}}ar, and Horesh]{Lu2021a}
Songtao Lu, Kaiqing Zhang, Tianyi Chen, Tamer Ba{\c{s}}ar, and Lior Horesh.
\newblock {Decentralized Policy Gradient Descent Ascent for Safe Multi-Agent
  Reinforcement Learning}.
\newblock \emph{35th AAAI Conference on Artificial Intelligence}, pages
  8767--8775, 2021.

\bibitem[Ma et~al.(2021)Ma, Chen, Eben, Lin, Guan, Ren, and Zheng]{Ma2021a}
Haitong Ma, Jianyu Chen, Shengbo Eben, Ziyu Lin, Yang Guan, Yangang Ren, and
  Sifa Zheng.
\newblock {Model-based Constrained Reinforcement Learning using Generalized
  Control Barrier Function}.
\newblock \emph{IEEE International Conference on Intelligent Robots and
  Systems}, pages 4552--4559, 2021.

\bibitem[Mesbah(2016)]{Mesbah2016}
Ali Mesbah.
\newblock {Stochastic model predictive control: An overview and perspectives
  for future research}.
\newblock \emph{IEEE Control Systems Magazine}, 36\penalty0 (6):\penalty0
  30--44, 2016.

\bibitem[Molina-Solana et~al.(2017)Molina-Solana, Ros, Ruiz,
  G{\'{o}}mez-Romero, and Martin-Bautista]{Molina-Solana2017}
Miguel Molina-Solana, María Ros, M.~Dolores Ruiz, Juan G{\'{o}}mez-Romero, and
  M.~J. Martin-Bautista.
\newblock {Data science for building energy management: A review}.
\newblock \emph{Renewable and Sustainable Energy Reviews}, 70:\penalty0
  598--609, 2017.

\bibitem[Papoudakis et~al.(2020)Papoudakis, Christianos, Sch{\"{a}}fer, and
  Albrecht]{Papoudakis2020}
Georgios Papoudakis, Filippos Christianos, Lukas Sch{\"{a}}fer, and Stefano~V.
  Albrecht.
\newblock {Benchmarking Multi-Agent Deep Reinforcement Learning Algorithms in
  Cooperative Tasks}.
\newblock In \emph{35th Conference on Neural Information Processing Systems},
  2020.

\bibitem[Parnika et~al.(2021)Parnika, Diddigi, Danda, and
  Bhatnagar]{Parnika2021}
P.~Parnika, Raghuram~Bharadwaj Diddigi, Sai Koti~Reddy Danda, and Shalabh
  Bhatnagar.
\newblock {Attention actor-critic algorithm for multi-agent constrained
  co-operative reinforcement learning}.
\newblock \emph{Proceedings of the International Joint Conference on Autonomous
  Agents and Multiagent Systems}, 3:\penalty0 1604--1606, 2021.

\bibitem[Paternain et~al.(2019)Paternain, Chamon, Calvo-Fullana, and
  Ribeiro]{Paternain2019}
Santiago Paternain, Luiz~F.O. Chamon, Miguel Calvo-Fullana, and Alejandro
  Ribeiro.
\newblock {Constrained reinforcement learning has zero duality gap}.
\newblock In \emph{Advances in Neural Information Processing Systems},
  volume~32, 2019.

\bibitem[Paternain et~al.(2022)Paternain, Calvo-Fullana, Chamon, and
  Ribeiro]{Paternain2022}
Santiago Paternain, Miguel Calvo-Fullana, Luiz~F.O. Chamon, and Alejandro
  Ribeiro.
\newblock {Safe Policies for Reinforcement Learning via Primal-Dual Methods}.
\newblock \emph{IEEE Transactions on Automatic Control}, 2022.

\bibitem[Rockafellar and Uryasev(2000)]{Rockafellar2000}
R.~Tyrrell Rockafellar and Stanislav Uryasev.
\newblock {Optimization of Conditional Value-at-Risk}.
\newblock \emph{Journal of Risk}, 2:\penalty0 21--42, 2000.

\bibitem[Rockafellar and Uryasev(2002)]{Rockafellar2002}
R.~Tyrrell Rockafellar and Stanislav Uryasev.
\newblock {Conditional value-at-risk for general loss distributions}.
\newblock \emph{Journal of Banking and Finance}, 26\penalty0 (7):\penalty0
  1443--1471, 2002.

\bibitem[Silver et~al.(2014)Silver, Lever, Heess, Degris, Wierstra, and
  Riedmiller]{Silver2014}
David Silver, Guy Lever, Nicolas Heess, Thomas Degris, Daan Wierstra, and
  Martin Riedmiller.
\newblock {Deterministic policy gradient algorithms}.
\newblock \emph{31st International Conference on Machine Learning, ICML 2014},
  1:\penalty0 605--619, 2014.

\bibitem[Tabas and Zhang(2022)]{Tabas2022}
Daniel Tabas and Baosen Zhang.
\newblock {Computationally Efficient Safe Reinforcement Learning for Power
  Systems}.
\newblock In \emph{Proceedings of the American Control Conference}, pages
  3303--3310. American Automatic Control Council, 2022.

\bibitem[Tessler et~al.(2019)Tessler, Mankowitz, and Mannor]{Tessler2019}
Chen Tessler, Daniel~J. Mankowitz, and Shie Mannor.
\newblock {Reward constrained policy optimization}.
\newblock In \emph{7th International Conference on Learning Representations,
  ICLR 2019}, May 2019.

\bibitem[Zhou et~al.(2022)Zhou, Chen, Yan, Li, Yin, and Ge]{Zhou2022a}
Wei Zhou, Dong Chen, Jun Yan, Zhaojian Li, Huilin Yin, and Wanchen Ge.
\newblock {Multi-agent reinforcement learning for cooperative lane changing of
  connected and autonomous vehicles in mixed traffic}.
\newblock \emph{Autonomous Intelligent Systems}, 2\penalty0 (1), 2022.

\end{thebibliography}

\includeappendix{

\newpage

\appendix

\section{Theoretical results} \label{appendix:thm}
\subsection{Proof of Proposition \ref{thm:10-17-22-1}}
Applying the definition of $\mu^\pi_\gamma$, we have $\int_{\X} \mu^\pi_\gamma(x)dx = \int_{\X} \G_{t=0}^\infty p_t^\pi(x) dx$. Using the Dominated Convergence Theorem, we can exchange the order of the sum and integral. Each individual $p_t^\pi$ integrates to $1$. The geometric sum property ensures that the resulting expression evaluates to $1$. 

\subsection{Proof of Proposition \ref{thm:10-17-22-2}}
\begin{enumerate}
    \item By definition, we have $\lim_{\gamma \rightarrow 0^+} \mu^\pi_\gamma(x) = 
    \lim_{\gamma \rightarrow 0^+} \G_{t=0}^\infty p_t^\pi(x)$. Using Tannery's theorem, we can exchange the order of the limit and the infinite sum. The zeroth term in the sum evaluates to $p_0(x)$ and all other terms evaluate to $0$. 
    \item Assume $\lim_{t \rightarrow \infty} p_t^\pi$ exists, and denote it $p_\infty^\pi$. Using the triangle inequality, we have \begin{align}
        |\mu^\pi_\gamma(x) - p_\infty^\pi(x)| &\leq \G_{t=0}^\infty |p_t^\pi(x) - p_\infty^\pi(x)| \\
        &= \G_{t=0}^{N} |p_t^\pi(x) - p_\infty^\pi(x)| + \G_{t=N+1}^\infty |p_t^\pi(x) - p_\infty^\pi(x)| \label{eqn:10-19-22-1}
    \end{align} for some $N \in \N$. Since $p_t^\pi(x) \rightarrow p_\infty^\pi(x)$, we can choose $N$ large enough to make the second term in \eqref{eqn:10-19-22-1} arbitrarily small. Then, using boundedness of $p_t^\pi$ for all $t$, we can take $\gamma \rightarrow 1^-$ to make the first term arbitrarily small.
\end{enumerate}

\subsection{Proof of Proposition \ref{thm:11-9-22-1}}
By the geometric sum property, we have
    $T_2(\gamma,\varepsilon) = \min\{K \in \N: \G_{t=0}^{K-1}[1] \geq 1 - \varepsilon\} = \min\{K \in \N: 1-\gamma^K \geq 1 - \varepsilon\} 
    = \min\{K \in \N: K \geq \frac{\log \varepsilon}{\log \gamma}\} = \big \lceil \frac{\log \varepsilon}{\log \gamma} \big \rceil.$
The termination time follows a geometric distribution with parameter $(1-\gamma)$, and thus has expected value $\frac{1}{1-\gamma}$. 
Setting $T_2(\gamma,\varepsilon) = T_1(\gamma)$ and solving for $\varepsilon$ (ignoring the integer constraint) yields $\varepsilon = \gamma^{\frac{1}{1-\gamma}}$. Finally, taking $\lim_{\gamma \rightarrow 1} \gamma^{\frac{1}{1-\gamma}}$ yields $\frac{1}{e}$.


\newversion{


\subsection{Proof of Proposition \ref{thm:11-14-22-1}}

Let $x \sim \mu^\pi_\gamma$, $x' \sim f^\pi(x)$, $\bar{c} = \E_{x \sim \mu^\pi_\gamma}[C(x)]$, and $\Sigma_C^2 = \E_{x \sim \mu^\pi_\gamma}[(\bar{c} - C(x))(\bar{c} - C(x))^T]$. Suppose $\lambda$ is randomly distributed with mean $\bar{\lambda}$ and variance $\Sigma_\lambda^2$. For any value function approximator $\hat{V}^i_\pi$, assume $\lambda$ and $\hat{V}^i_\pi$ are independent. Let $\eta = \m{1 & \lambda^T}^T$, $\ d = \m{R^i(x) & C(x)^T}^T$,
$\hat{V}_\pi^i: \X \rightarrow \R,\ \hat{V}_{R,\pi}^i: \X \rightarrow \R,$ and $\hat{V}_{C,\pi}: \X \rightarrow \R^m$. Let $\D$ be a dataset of trajectories sampled from $\M$ that is used to train $\hat{V}_\pi^i,\ \hat{V}_{R,\pi}^i$, and $\hat{V}_{C,\pi}$. The mean square temporal difference error achieved by using a generic value function is \begin{align}
    MSTDE_1 = \E_{x,x',\lambda,\D}[(\eta^Td + \gamma \hat{V}_\pi^i(x') - \hat{V}_\pi^i(x))^2] \label{eqn:11-14-22-1}
\end{align} while the error achieved using the structured value function is \begin{align}
    MSTDE_2 = \E_{x,x',\D}[(\eta^Td + \gamma[\hat{V}_{R,\pi}^i(x') - \lambda^T \hat{V}_{C,\pi}(x')] - [\hat{V}_{R,\pi}^i(x) - \lambda^T \hat{V}_{C,\pi}(x)])^2]. \label{eqn:11-14-22-2}
\end{align} Note that in \eqref{eqn:11-14-22-2} we do not take the expectation over $\lambda$ since the dual variables are made available to this function approximator.

Begin with the states and dual variables fixed at $(\bar{x},\bar{x}',\bar{\lambda})$. Let $\hat{g}(\bar{x},\bar{x}') = \m{\hat{V}_{R,\pi}^i(\bar{x}) & \hat{V}_{C,\pi}(\bar{x})^T}^T - \gamma \m{\hat{V}_{R,\pi}^i(\bar{x}') & \hat{V}_{C,\pi}(\bar{x}')^T}^T$ and $\hat{h}(\bar{x},\bar{x}') = \hat{V}_\pi^i(\bar{x}) - \gamma \hat{V}_\pi^i(\bar{x}')$. Then, suppressing the arguments $(\bar{x},\bar{x}')$ and setting $\bar{\eta} = \m{1 & -\bar{\lambda}^T}^T$, we can write the squared temporal difference error at $(\bar{x},\bar{x}',\bar{\lambda})$ as \begin{align}
    STDE_1(\bar{\eta}) = \E_\D[(\bar{\eta}^Td - \hat{h})^2],\\
    STDE_2(\bar{\eta}) = \E_\D[(\bar{\eta}^Td - \bar{\eta}^T \hat{g})^2]. \label{eqn:11-14-22-5}
\end{align}
The loss function used to train $\hat{V}_{R,\pi}^i$ and $\hat{V}_{C,\pi}$ is \begin{align}
    \E_{\D}[\|d - \hat{g}\|^2]. \label{eqn:11-14-22-3}
\end{align} Since $d$ is a deterministic function of $x$, \eqref{eqn:11-14-22-3} can be decomposed into bias and variance terms: \begin{align}
    \E_{\D}[\|d - \hat{g}\|^2] &= \E_\D[\sum_{k=0}^m (d_k - \hat{g}_k)^2] \\
    &= \sum_{k=0}^m \E_\D[(d_k - \hat{g}_k)^2] \\
    &= \sum_{k = 0}^m [(d_k - \E_\D \hat{g}_k)^2 + \E_\D[(\hat{g} - \E_\D \hat{g})^2]] \\
    &:= \sum_{k=0}^m[b_k^2 + \sigma_k^2] \\
    &:= \textup{Tr}[bb^T + \Sigma^2] \label{eqn:11-14-22-4}
\end{align} where $k = 0$ corresponds to the reward signal and $k = 1,\ldots,m$ corresponds to the cost signals.

Following a similar line of reasoning, we can use \eqref{eqn:11-14-22-4} to rewrite \eqref{eqn:11-14-22-5} as \begin{align}
    STDE_2(\bar{\eta}) = \textup{Tr}[(bb^T + \Sigma^2)(\bar{\eta} \bar{\eta}^T)].
\end{align}

For the sake of argument, we assume that $\hat{g}$ and $\hat{h}$ achieve the same performance at $(x,x',\lambda)$, that is, \begin{align}
    STDE_1(\bar{\eta}) = STDE_2(\bar{\eta}) = \textup{Tr}[(bb^T + \Sigma^2)(\bar{\eta}\bar{\eta}^T)]
\end{align} where $\textup{Tr}[(bb^T)(\bar{\eta}\bar{\eta}^T)]$ and $\textup{Tr}[\Sigma^2 \bar{\eta}\bar{\eta}^T]$ reflect the bias squared and variance terms, respectively.
How do $STDE_1$ and $STDE_2$ change when $\lambda$ is allowed to vary? Using the generic estimator, the noise in $\lambda$ will introduce some amount of \textit{irreducible error} into $STDE_1$. On the other hand, using $\lambda = \bar{\lambda} + \Delta \lambda$ in our proposed estimator will change the bias and variance terms in $STDE_2$ while the irreducible error remains at zero (since there is no uncertainty when $\Delta \lambda$ is known). Setting $\Delta \eta = \m{0 & -\Delta \lambda^T}^T,$ the temporal difference errors at $(\bar{x},\bar{x}',\bar{\lambda} + \Delta \lambda)$ are \begin{align}
    STDE_1(\bar{\eta} + \Delta \eta) = \textup{Tr}[(bb^T + \Sigma^2)(\bar{\eta} \bar{\eta}^T)] + (\Delta \eta^Td)^2, \\
    STDE_2(\bar{\eta} + \Delta \eta) = \textup{Tr}[(bb^T + \Sigma^2)((\bar{\eta}+\Delta \eta) (\bar{\eta} + \Delta \eta)^T)].
\end{align}

Taking the expectation over $\Delta \lambda$ which has a mean of zero and a variance of $\Sigma_\lambda^2$, and setting $\Sigma_\eta^2 = \m{0 & 0 \\ 0 & \Sigma_\lambda^2}$, yields \begin{align}
    \E_{\Delta \lambda}[STDE_1(\bar{\eta} + \Delta \eta) - STDE_2(\bar{\eta} + \Delta \eta)] &= \textup{Tr}[\Sigma_\eta^2(dd^T - bb^T - \Sigma^2)] \\
    &= \textup{Tr}[\Sigma_\lambda^2(cc^T - \tilde{b}\tilde{b}^T - \tilde{\Sigma}^2)] \label{eqn:11-14-22-6}
\end{align} where $\tilde{b} = (c - \E_\D \hat{g}_C)$, $\tilde{\Sigma}^2 = \E_\D[(\hat{g}_C - \E_\D\hat{g}_C)^2],$ and $\hat{g}_C = \hat{V}_{C,\pi}(x) - \gamma \hat{V}_{C,\pi}(x')$.
Note that $\E_\D[\|c - \hat{g}_C\|^2] = \textup{Tr}[\tilde{b}\tilde{b}^T + \tilde{\Sigma}^2].$ Taking $\tilde{b},\tilde{\Sigma}^2 \rightarrow 0$ as the accuracy of $\hat{g}_C$ improves, \eqref{eqn:11-14-22-6} can be estimated as \begin{align}
    \textup{Tr}[\Sigma_\lambda^2 cc^T]. \label{eqn:11-29-22-1}
\end{align} Taking the expectation over $c \sim C(x), x \sim \mu^\pi_\gamma$ yields the final result.

}{}

\section{Simulation details} \label{appendix:simulations}
\subsection{Algorithm}

The Constrained Multiagent Advantage Actor Critic (C-MAA2C) algorithm is shown in Algorithm \ref{alg:MAA2C}. The main differences from the basic MAA2C algorithm are the penalty modifications in lines \ref{line:chance} and \ref{line:cvar}, the use of vector-valued value functions $\hat{V}^i : \X \rightarrow \R^{m+1}$ (one per agent in the noncooperative setting), and the dual update.

There are two apparent differences between Algorithm \ref{alg:MAA2C} and the concepts described in the main text. The first is that Algorithm \ref{alg:MAA2C} uses n-step returns in the advantage function, whereas Section \ref{sec:value} only considers one-step returns. We resolve this discrepancy by revisiting the proof of Proposition \ref{thm:11-14-22-1}. First, note that the coefficients $\eta$ can be factored out of the returns just like they are factored out of the rewards. Thus, the proof only requires slight modifications up to the last line, Equation \eqref{eqn:11-29-22-1}. Using returns instead of rewards in \eqref{eqn:11-29-22-1} will lead to a different numerical result but the conclusion (justification for using a structured value function) will be the same. 

The second apparent difference is the fact that Algorithm \ref{alg:MAA2C} considers finite-horizon episodic tasks, thus the primal-dual algorithm will enforce $\E_{\tau \sim \M}[\G_{t=0}^T c_t] \leq 0$. Due to the finite horizon, we cannot directly use the occupation measure to interpret the meaning of this constraint. However, we can define the occupation measure over a finite horizon as \begin{align}
    \mu^\pi_{\gamma,T}(x) = \frac{1}{1-\gamma^{T+1}}\G_{t=0}^{T} p_t^\pi(x).
\end{align} It is easy to show that $\mu^\pi_{\gamma,T}$ is nonnegative and integrates to unity over $\X$. We can use $\mu^\pi_{\gamma,T}$ in place of $\mu^\pi_\gamma$ everywhere in order to interpret discounted sum constraints and to generate probabilistic constraints in finite-horizon episodic tasks. The statements $\E_{\tau \sim \M}[\G_{t=0}^T c_t] \leq 0,\ \E_{\tau \sim \M}[(1-\gamma^{T+1})^{-1}\G_{t=0}^T c_t] \leq 0,$ and $\E_{x \sim \mu^\pi_{\gamma,T}}[C(x)] \leq 0$ are equivalent. Note that the effective horizon discussed in Section \ref{sec:occupation} may be shorter than the horizon length $T$. 

\begin{algorithm}[ht]
	\caption{C-MAA2C with probabilistic safety guarantees and structured value functions}
	\label{alg:MAA2C}
	\begin{algorithmic}[1]
	    \STATE Input discount factor $\gamma$, learning rates $\zeta_\theta,\ \zeta_\omega,\ \zeta_\lambda$, n-step return horizon $\kappa$, tolerances $\alpha$ and $\delta$, multiplier limit $\lambda_{\text{max}}$, episode length $T$, number of episodes $K$, risk metric $\in \{\text{average, chance, CVaR}\}$ \\
	    \STATE Initialize actor parameters $\{\theta^i\}_{i\in \n}$, critic parameters $\{\omega^i\}_{i \in \n}$, parameterized policies $\pi^i(\cdot \mid \theta^i): \X_i \rightarrow \Delta_{\U_i}$, parameterized value estimates $\hat{V}^i(\cdot \mid \omega^i): \X \rightarrow \R^{m+1},$ dual variables $\lambda \in \R^m$, $\eta := \m{1 & -\lambda^T}^T$ 
	    \FOR{$k = 1,2,\ldots,K$}
	        \STATE Initialize $x_0 \sim p_0$ \COMMENT{Run 1 episode}
	        \FOR{$t = 0,1,\ldots T$}
	            \STATE Sample $u_t^i \sim \pi^i(\cdot \mid x_t^i, \theta^i)$ for $i \in \n$
	            \STATE Receive $\{r_t^i\}_{i \in \n},\ c_t, x_{t+1}$
	            \IF{risk metric = chance}
	                \STATE $c_t \leftarrow I[c_t \geq \alpha] - \delta$ \COMMENT{Proposition \ref{thm:10-18-22-1}} \label{line:chance}
	            \ELSIF{risk metric = CVaR}
	                \STATE $c_t \leftarrow [c_t - \alpha]_+ - \delta$ \COMMENT{Proposition \ref{thm:10-20-22-1}} \label{line:cvar}
	            \ENDIF
	            \STATE Let $d_t^i = \m{r_t^i & c_t^T}^T$ for $i \in \n$
	        \ENDFOR
	        \FOR{$i \in \n$}
	            \FOR{$t = 0,1,\ldots, T$}
    	            \STATE $N = \min\{T,t + \kappa\}$
    	            \STATE $D_t^i = \sum_{n=t}^{N-1} \gamma^{n-t}d_n^i + \gamma^{N-t}\hat{V}^i(x_{N}\mid \omega^i)$ \COMMENT{Compute n-step returns}
    	            \STATE $A_t^i = \eta^T(D_t^i - \hat{V}^i(x_t\mid\omega^i))$ \COMMENT{Compute advantages}
	            \ENDFOR
	            \STATE $\theta^i \leftarrow \theta^i + \zeta_\theta \sum_{t=0}^T A_t^i \nabla_{\theta^i} \log \pi^i(u_t^i\mid x_t^i,\theta^i)$ \COMMENT{Actor update}
	            \STATE $\omega^i \leftarrow \omega^i - \zeta_\omega \nabla_{\omega^i}\sum_{t=0}^T\|D_t^i - \hat{V}^i(x_t \mid \omega^i)\|_2^2$ \COMMENT{Critic update}
	        \ENDFOR
	        \STATE $\lambda \leftarrow \lambda + \zeta_\lambda \G_{t=0}^T c_t$ \COMMENT{Dual update}
	        \STATE $\lambda \leftarrow \min\{[\lambda]_+,\lambda_{\text{max}}\}$
	        \STATE $\eta \leftarrow \m{1 & -\lambda^T}^T$ 
	    \ENDFOR
	\end{algorithmic}
\end{algorithm}

\newpage

\subsection{Hyperparameters}

Simulation hyperparameters are listed in Table \ref{table:hparams}. 

\begin{table*}[ht]
	\vskip 0.15in
	\begin{center}
		\begin{tabular*}{\textwidth}{lc}
		    \hline
			\hline
			\textbf{Simulation} \\
			\quad Episode length \hspace{19em} & $25$\\
			\quad Number of episodes & $\{4,8\} \times 10^4$ \\
			\quad Number of trials per configuration & 5 \\
			\hline
			\textbf{RL algorithm} \\
			\quad Discount factor $\gamma$ & 0.99\\
			\quad Actor learning rate $\zeta_\theta$ & $3 \times 10^{-4}$ \\
			\quad Critic learning rate $\zeta_\omega$ & $3 \times 10^{-4}$ \\
			\quad Dual update step size $\zeta_\lambda$ & $1 \times 10^{-4}$ \\
			\quad Optimizer & Adam$(\beta_{\text{Adam}} = (0.9, 0.999))$\\
			\quad n-step return horizon $\kappa$ & 5\\
			\hline
			\textbf{Constraint enforcement}\\
			\quad $\lambda_{\text{max}}$ & 10\\
			\quad Risk level $\beta$ & 0.9\\
			\quad ``LHS tolerance'' $\alpha$:\\
			\quad \quad Average constraints & N/A \\
			\quad \quad Chance constraints & 0.1 \\
			\quad \quad CVaR constraints & 0.2 \\
		    \quad ``RHS tolerance'' $\delta$:\\
		    \quad \quad Average constraints & 0\\
		    \quad \quad Chance constraints & 0.1 \\
		    \quad \quad CVaR constraints & $5 \times 10^{-3}$\\
			\hline
			\textbf{Actors} \\
			\quad Policy architecture & Multi-layer perceptron \\
			\quad Number of hidden layers & 2 \\
			\quad Hidden layer width & 64 \\
			\quad Hidden layer activation & ReLU\\
			\quad Output layer activation & Linear \\
			\quad Action selection & Categorical sampling\\
			\quad Parameter sharing & No \\
			\hline
			\textbf{Critics} \\
			\quad Critic architecture & Multi-layer perceptron \\
			\quad Number of hidden layers & 2 \\
			\quad Hidden layer width & 64 \\
			\quad Hidden layer activation & ReLU\\
			\quad Output layer activation & Linear \\
			\quad Target network update interval & 200 episodes \\
			\quad Parameter sharing & No \\
		    \hline
		    \hline
		\end{tabular*}
	\end{center}
	\vskip -0.1in
	\caption{Simulation hyperparameters.}
	\label{table:hparams}
\end{table*}

\newpage

\subsection{Additional simulation results}

Here, we provide some additional results to supplement the findings in Section \ref{sec:simulations}. First, we compared the convergence to a safe policy under the original discounted sum constraint and found that similar to the results for the other types of constraints, the structured critic demonstrates a better safety margin throughout training. This is illustrated in Figure \ref{fig:sim_avg}. 

\begin{figure}[h]
    \centering
    \includegraphics[width=0.4\textwidth]{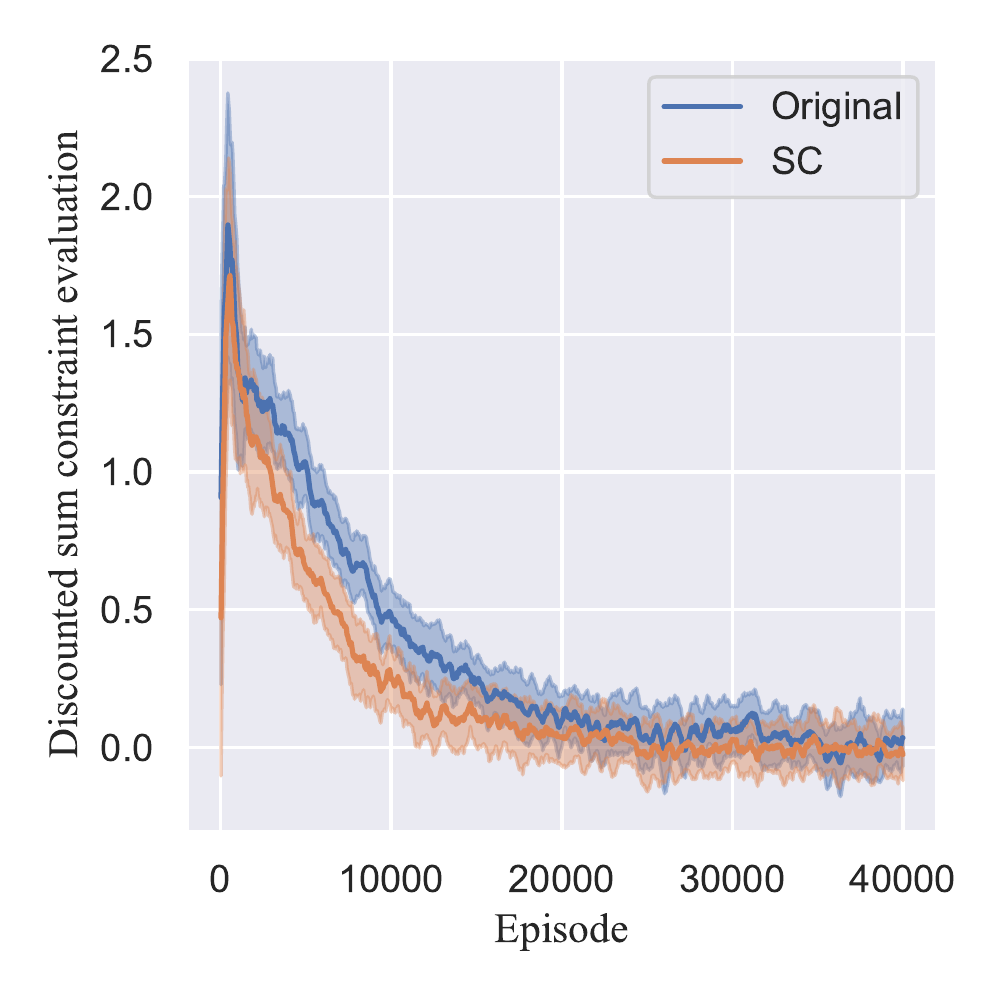}
    \caption{Evaluation of the discounted sum constraint throughout training, showing that the structured critic helps the actor to find safer policies faster. Each line and shaded region represents the mean and standard deviation over 5 training runs. Key: SC = structured critic.}
    \label{fig:sim_avg}
\end{figure}

Next, we provide a closer look at the accuracy of the CVaR upper bound provided in Proposition \ref{thm:10-20-22-1}, and illustrated using dashed lines in the left panel of Figure \ref{fig:sim_cvar}. Table \ref{table:cvar} shows that in all four configurations in which the CVaR was evaluated, the upper bound is a fairly accurate estimate. The results from Section \ref{sec:simulations} show that this upper bound can be used to drive the actual CVaR below a target value. Although using a structured critic with modified penalty function yielded the most accurate CVaR UB, the accuracy in all four configurations could be improved by making further adjustments to the tolerance $\alpha$. The error is reported for policies tested at the end of the training phase. 

\begin{table*}[ht]
	\vskip 0.15in
	\begin{center}
		\begin{tabular}{ccc}
		    \hline
			\hline
			Penalty function & Critic & CVaR UB error\\
			\hline
			$C(x)$ & Generic & 18.3\%\\
			$C(x)$ & Structured & 11.8\%\\
			$[C(x)-\alpha]_+ - \delta$ & Generic & 7.6\%\\
			$[C(x)-\alpha]_+ - \delta$ & Structured & 3.7\%\\
		    \hline
		    \hline
		\end{tabular}
	\end{center}
	\vskip -0.1in
	\caption{Accuracy of CVaR upper bound.}
	\label{table:cvar}
\end{table*}}{}

\end{document}